\newif\ifllncs
\newif\ifsoda
\sodafalse
\newif\ifhasappendix
\hasappendixfalse

\ifllncs
\hasappendixtrue
\fi

\ifllncs
\let\accentvec\vec  
\documentclass 
{llncs}
\let\vec\accentvec 
\else
\documentclass[a4paper,11pt]{article}
\fi

\usepackage{etex,etoolbox}

\usepackage[T1]{fontenc}

\usepackage[ngerman,american]{babel} 

\usepackage{amsmath}
\usepackage{amsfonts}
\usepackage{amssymb}
\usepackage{graphicx}
\usepackage{color}
\usepackage{tikz}

\usepackage{xspace}

\ifllncs
\newcommand{\scalefactor}{0.75}
\else
\newcommand{\scalefactor}{1}
\fi

\ifllncs

\spnewtheorem*{Proof}{Proof}{\itshape}{\rmfamily}
\renewenvironment{proof}{\begin{Proof}}{\qed\end{Proof}}
\else
\usepackage{fullpage}
\usepackage{amsthm}
\theoremstyle{plain}
\newtheorem{theorem}{Theorem}
\newtheorem{lemma}[theorem]{Lemma}
\newtheorem{claim}[theorem]{Claim}
\newtheorem{corollary}[theorem]{Corollary}
\theoremstyle{definition}

\fi

\ifhasappendix
\usepackage{multibib}
\newcites{app}{References for the Appendix}
\fi

\def\appcite#1{\ifhasappendix\citeapp{#1}\else\cite{#1}\fi}

\makeatletter
\providecommand{\@fourthoffour}[4]{#4}  
\ifllncs
\providecommand{\@motsacces}[6]{#3} 
\else
\providecommand{\@motsacces}[4]{#4} 
\fi 

\def\fixstatement#1{%
  \AtEndEnvironment{#1}{%
 
       \xdef\pat@label{\expandafter\expandafter\expandafter
            \@motsacces\csname#1\endcsname~\space \@currentlabel}}}
\globtoksblk\prooftoks{1000}
\newcounter{proofcount}

\long\def\proofatend#1\endproofatend{%
\ifhasappendix
\marginpar{\scriptsize{\vspace{-1cm}{\color{gray}{Proof in Appendix}}}}
  \edef\next{\noexpand \begin{proof}[Proof of \pat@label]
  }%
  \toks\numexpr\prooftoks+\value{proofcount}\relax=\expandafter{\next#1
  \end{proof}
  }
  \stepcounter{proofcount}
\else
\begin{proof}
#1
\end{proof}
\fi
}
\long\def\differingappendixstatement#1#2{
\ifhasappendix
\edef\next{}%
  \toks\numexpr\prooftoks+\value{proofcount}\relax=\expandafter{\next#2}
  \stepcounter{proofcount}
 \else
 #1
\fi
}

\long\def\hereorinappendixstatement#1{%
\ifhasappendix%
\edef\next{}%
  \toks\numexpr\prooftoks+\value{proofcount}\relax=\expandafter{\next#1}%
  \stepcounter{proofcount}%
\else%
#1%
\fi%
}

\def\printproofs{%
  \count@=\z@
  \loop
    \the\toks\numexpr\prooftoks+\count@\relax
    \ifnum\count@<\value{proofcount}%
    \advance\count@\@ne
  \repeat}
\makeatother
\fixstatement{theorem}
\fixstatement{lemma}
\fixstatement{corollary}
\fixstatement{definition}

\usepackage[colorlinks]{hyperref}

\definecolor{lightblue}{rgb}{0.5,0.5,1.0}
\definecolor{darkred}{rgb}{0.5,0,0}
\definecolor{darkgreen}{rgb}{0,0.5,0}
\definecolor{darkblue}{rgb}{0,0,0.5}

\hypersetup{colorlinks,linkcolor=darkred,filecolor=darkgreen,urlcolor=darkred,citecolor=darkblue}

\usepackage{algorithm}
\usepackage{algorithmic}
\newcommand{\ENSUREGAP}{\vspace{0.2cm}}
\definecolor{gray}{gray}{0.3}

\makeatletter
\providecommand*{\toclevel@algorithm}{0}
\makeatother

\tikzstyle{vertex}=[circle,fill=white,draw,very thick, minimum size=11pt,inner sep=1pt]
\tikzstyle{grayvertex}=[circle,fill=gray,draw,very thick, minimum size=11pt,inner sep=1pt]

\newcommand{\CClassNP}{\textup{NP}\xspace}

\newcommand{\NPhard}{\CClassNP-hard\xspace}

\newcommand{\BigO}{\mathcal{O}}

\DeclareMathOperator{\wl}{wl}

\DeclareMathOperator{\isoalgo}{ISO}

\DeclareMathOperator{\tdw}{tdw}
\DeclareMathOperator{\stw}{stw}

\DeclareMathOperator{\ctdw}{rctdw} 
\DeclareMathOperator{\rtdw}{rtdw} 
\DeclareMathOperator{\cstw}{cstw} 
\DeclareMathOperator{\isored}{Red}

\DeclareMathOperator{\kcon}{\equiv}
\DeclareMathOperator{\Orc}{Orc}

\title{Reduction Techniques for Graph Isomorphism in the Context of Width Parameters}
\ifllncs
\author{Yota Otachi\inst{1} \and Pascal Schweitzer\inst{2} 
}
\institute{Japan Advanced Institute of Science and Technology\\ School of Information Science,  {\tt otachi@jaist.ac.jp} \\ \and
 RWTH Aachen University, \\{\tt schweitzer@informatik.rwth-aachen.de} 
}
\else
\author{Yota Otachi and Pascal Schweitzer\thanks{This work is partially supported by Grant-in-Aid for Young
Scientists (B) 25730003, Japan Society for the Promotion of Science.}\\
\\ Japan Advanced Institute of Science and Technology\\ School of Information Science\\ 
Asahidai 1-1, Nomi, Ishikawa 923-1292, Japan\\
{\tt otachi@jaist.ac.jp}\\ \\
 RWTH Aachen University\\
 Ahornstra\ss{}e 55, 52074 Aachen, Germany
 \\{\tt schweitzer@informatik.rwth-aachen.de} 
}
\fi

\begin{document}

\maketitle

\begin{abstract}
We study the parameterized complexity of the graph isomorphism problem when parameterized by width parameters related to tree decompositions. We apply the following technique to obtain fixed-parameter tractability for such parameters. We first compute an isomorphism invariant set of potential bags for a decomposition and then apply a restricted version of the Weisfeiler-Lehman algorithm to solve isomorphism. With this we show fixed-parameter tractability for several parameters and provide a unified explanation for various isomorphism results concerned with parameters related to tree decompositions.

As a possibly first step towards intractability results for parameterized graph isomorphism we develop an fpt Turing-reduction from  strong tree width to the a priori unrelated parameter maximum degree.
\end{abstract}

\section{Introduction}

The graph isomorphism problem is the algorithmic task to decide whether two given graphs are isomorphic, i.e., whether there exists a bijection from the vertices of one graph to the vertices of the other graph preserving adjacency and non-adjacency. The problem is situated in the complexity class \CClassNP. However, despite extensive research on this problem, the complexity remains unknown. It is neither known whether the problem is polynomial-time solvable nor whether it is \NPhard. 

In this paper, we are interested in the parameterized complexity of the isomorphism problem. For other aspects related to the isomorphism problem we refer the reader to other sources (e.g.,~\cite{KoeblerSchoeningToran:1993},~\cite{SchweitzerThesis},~\cite{Toda:2001:GIDP}).

In the parameterized context, for a graph parameter~$k$, such as the maximum degree of the input graphs, we ask for an algorithm that solves isomorphism of graphs with parameter at most~$k$. In this context, we are interested in the existence of algorithms with a running time of~$O(f(k) n^c)$ for some constant~$c\in \mathbb{N}$ in contrast to algorithms with a running time of~$O(n^{f(k)})$. Running times of the former type are called fpt time and the algorithms are said to be fixed-parameter tractable algorithms.

\emph{Related work.} There are various results that show that isomorphism is fixed-parameter tractable with respect to some parameter. Such results exist for the parameters color multiplicity~\cite{DBLP:conf/focs/FurstHL80} (also known for hypergraphs~\cite{DBLP:conf/fsttcs/ArvindDKT10}), eigenvalue multiplicity~\cite{Evdokimov}, rooted distance width~\cite{DBLP:journals/algorithmica/YamazakiBFT99}, feedback vertex set number~\cite{DBLP:conf/swat/KratschS10}, bounded permutation distance~\cite{DBLP:conf/esa/Schweitzer11}, tree-depth~\cite{DBLP:conf/iwpec/BoulandDK12} and connected path distance width~\cite{Otachi2012ISAAC}.
For chordal graphs, tractability results are known for the parameters clique number~\cite{doi:10.1137/0603025},~\cite{DBLP:conf/isaac/Nagoya01} and the size of simplicial components~\cite{DBLP:journals/ieicet/Toda06}. Yet, for many parameters, such as maximum degree, tree width and genus, it is not known whether there exist fixed-parameter tractable algorithms solving  isomorphism (see~\cite{DBLP:conf/swat/KratschS10}).
However, no non-tractability results are known. One of the obstacles to understanding the parameterized complexity of graph isomorphism is the uncertainty whether the standard reduction techniques, like showing W[1]-hardness, can be applied (see the discussions in~\cite{DBLP:conf/swat/KratschS10} and~\cite{DBLP:journals/algorithmica/YamazakiBFT99}).

\emph{Our results.}  We study the parameterized complexity of isomorphism with respect to various parameters related to strong tree decompositions.
We first develop a method to obtain fixed-parameter tractable algorithms for parameterized graph isomorphism problems.
The underlying technique of many results showing such results is to first find a restricted isomorphism invariant family of sets, potential bags, which capture a tree decompositions and to then use these to perform an isomorphism test that uses some form of dynamic programming.
It turns out that it is possible to prove that this technique is applicable in general. To prove this general statement, we develop a restricted version of the Weisfeiler-Lehman color refinement algorithm and prove that it successfully decides isomorphism whenever an invariant family of potential bags capturing tree decompositions is available for the input graphs. The algorithm neither computes a  decomposition nor is it required that a decomposition is given to the algorithm.

Using the technique, we show tractability of graph isomorphism for the parameters
root-connected tree distance width and connected strong tree width.  We also provide families of examples showing that neither of the two graph parameters mentioned in the theorem can be bounded by a function of the other. The two tractability results extend results in~\cite{fuhlbrueck},~\cite{Otachi2012ISAAC}, and~\cite{DBLP:journals/algorithmica/YamazakiBFT99} also concerned with restricted forms of strong tree decompositions, also answering a question from~\cite{DBLP:journals/algorithmica/YamazakiBFT99}.
Furthermore, with the technique, it is for example also possible to show that graph isomorphism parameterized by the maximum of the length of a longest geodesic cycle and strong tree width or by the maximum of the chordality and degree is fixed-parameter tractable. 

In general, our technique provides a unified explanation for the various results~\cite{DBLP:conf/iwpec/BoulandDK12,fuhlbrueck,doi:10.1137/0603025,DBLP:conf/swat/KratschS10,DBLP:conf/isaac/Nagoya01,Otachi2012ISAAC,DBLP:journals/ieicet/Toda06,DBLP:journals/algorithmica/YamazakiBFT99} all showing that certain restrictions on tree decompositions lead to efficient algorithms for the isomorphism problem. Indeed, all of these approaches can be interpreted as 
determining some restricted family of potential bags capturing a tree decomposition
and then performing some form of dynamic programming to check for isomorphism, that can also be performed by the restricted Weisfeiler-Lehman algorithm. In
each of the references above, 
the dynamic programming is a substantial part of the argumentation, which can  
now be replaced by the general theorem.

Finally, we show how the technique can be applied to obtain parameterized isomorphism algorithm by exploiting knowledge on the set of potential maximal cliques, of which we already know that it can always be computed in polynomial time in the number of potential maximal cliques.

Our technique also provides a proof of the fact that for graphs of bounded tree width a sufficiently high-dimensional Weisfeiler-Lehman algorithm can be used to determine isomorphism. This fact was first proven by Grohe and Mari{\~n}o using logic~\cite{DBLP:conf/icdt/GroheM99} (see also~\cite{DBLP:conf/lics/Grohe10}) and provides to date the fastest running time for isomorphism of bounded tree width graphs. Our proof provides a direct argument for this fact, which does not involve logic. We remark that in his book, Toda~\cite{Toda:2001:GIDP} also gives a dynamic programming algorithm matching the running time of the algorithm of Grohe and Mari{\~n}o.

In this paper, we also take a first step towards developing means for some form of intractability result. Specifically, for the isomorphism problem we construct an fpt Turing reduction from  strong tree width to the a priori unrelated parameter maximum degree. The existence of this reduction in particular implies that if graph isomorphism is fixed-parameter tractable when parameterized by degree then it is also fixed-parameter tractable when parameterized by strong tree width. 
However, a possibly better interpretation of this result is that isomorphism parameterized by degree is hard, being at least as intractable as isomorphism parameterized by strong tree width.

To obtain the reduction, we reduce the problem to biconnected components, a technique frequently used for isomorphism algorithms concerned with planar graphs (see~\cite{DBLP:conf/coco/DattaLNTW09}). However, we require an extended form of such a reduction allowing us to work with graphs equipped with an equivalence relation and equipped with a coloring of the linear orders of the equivalence classes.

\section{Preliminaries}
In this paper all graphs are finite, simple, undirected graphs. A \emph{biconnected component} (also called a block) is a maximal connected subgraph not containing a cut-vertex. In particular, the connected graph on 2 vertices is biconnected.

A \emph{strong tree decomposition} of a graph $G = (V,E)$
is a pair $(\{X_{i} \mid i \in I\}, T = (I,F))$ where~$\{X_{i} \mid i \in I\}$ is a partition of the vertex set~$V$ into so-called bags~$X_i$ and~$T=(I,F)$ is a tree such that the following holds: for all edges $\{u,v\} \in E$, either there is $i \in I$ with $u,v \in X_{i}$, or there are two adjacent tree vertices $i,i' \in I$ such that $u \in X_{i}$ and $v \in X_{i'}$.
A \emph{connected strong tree decomposition} is a strong tree decomposition for which each bag~$X_{i}$ induces a connected subgraph. The \emph{width} of a strong tree decomposition is the maximum size of a bag of the decomposition.

A strong tree decomposition $(\{X_{i} \mid i \in I\}, T = (I,F))$ with a distinguished root~$r \in I$
is a \emph{tree distance decomposition} if each $v \in X_{i}$ with $i \ne r$ has a neighbor $u \in X_{j}$
where $j$ is the parent of $i$ in $T$ rooted at $r$.
A tree distance decomposition with root $r$ is a \emph{root-connected tree distance decomposition} if $X_{r}$ induces a connected subgraph.

Here, we slightly diverge from the terminology used in~\cite{Otachi2012ISAAC} to highlight the fact that only the root set must induce a connected graph, and thereby avoid confusion with the term connected strong tree decomposition.

For a class of decompositions~$\mathcal{C}$, the $\mathcal{C}$ width of a graph~$G$ is the minimal width over all~$\mathcal{C}$ decompositions of~$G$. We thus obtain the 
graph parameters strong tree width, denoted~$\stw(G)$, connected strong tree width, denoted~$\cstw(G)$, tree distance width, denoted~$\tdw(G)$ and root-connected tree distance width, denoted $\ctdw(G)$.
The notion of strong tree width was introduced by Seese~\cite{Seese:1985:TGC:647892.739763} and is also known as \emph{tree-partition width}~\cite{DBLP:journals/dm/DingO96}. 
In the context of graph isomorphism, tree distance decompositions were first considered in~\cite{DBLP:journals/algorithmica/YamazakiBFT99}. 

For a graph~$G$, there may be several tree distance decompositions with the same root set~$S$.
However, there is a unique minimal decomposition (i.e., the partition into bags at least as fine as any other partition into bags obtained from a tree distance decomposition) with root set~$S$. Given~$S$, this minimal decomposition can be computed linear time. 

\begin{theorem}[{\cite[Theorem 2.1]{DBLP:journals/algorithmica/YamazakiBFT99}}]\label{thm:minimal:dist:decomp} Given a graph~$G$ and a set~$S$, one can compute in~$O(m)$ time the unique tree distance decomposition with root set~$S$.
\end{theorem}

 We denote the width of this decomposition by~$\tdw_S(G)$.
Note that if~$G$ is not connected, it may be the case that there is no tree distance decomposition with root set~$S$.
To facilitate our proofs and simplify algorithms, we define~$\tdw_S(G)$ to be infinite in this case.

For a graph~$G$ with distinct non-adjacent vertices~$s$ and~$t$ an~\emph{$s$-$t$-separator} is a set of vertices~$S$ such that~$s$ and~$t$ are in different components of~$G-S$. An $s$-$t$ separator is minimal if no proper subset of~$S$ is an~$s$-$t$-separator.

An \emph{fpt Turing reduction} (see \cite{Flum2006}) of a parameterized problem~$P_1$ with parameter~$k_1$ to a parameterized problem~$P_2$ with parameter~$k_2$ is a Turing reduction from~$P_1$ to~$P_2$ with fpt running time for which the parameter~$k_2$ of all oracle calls to the problem~$P_2$ is bounded by a computable function in terms of~$k_1$. In other words, a turing reduction is an fpt-algorithm solving the parameterized problem~$P_1$ with the help of an oracle that solves problem~$P_2$ such that there exists a computable function~$g$ such that for all oracle queries~$y\in P_2$ posed on an input~$x$ with parameter~$k_1$ it holds that the parameter~$k_2$ of~$y$ is at most~$g(k_1)$.

Suppose we assign to every graph~$G$ a subset of the vertices~$\mathcal{V}(G) \subseteq V(G)$. We say this assignment is \emph{isomorphism invariant} if for every isomorphism~$\pi \colon G_1 \rightarrow G_2$ we have~$\mathcal{V}(G_2) = \pi(\mathcal{V}(G_1))$. This definition extends to assignments of tuples or sets of vertex sets and also to colored graphs.

\section{Tree decompositions and the Weisfeiler-Lehman algorithm}\label{sec:weisfeiler:lehman:things}
\differingappendixstatement{}{\section{Proofs of Section~\ref{sec:weisfeiler:lehman:things}}}

In the graph isomorphism literature, for various graph classes, results are known showing that the Weisfeiler-Lehman algorithm yields polynomial time isomorphism algorithms (see \cite{DBLP:conf/lics/Grohe10}). In this section we describe a restricted version of the Weisfeiler-Lehman algorithm and show that it can be used to obtain fixed-parameter tractability results. Intuitively, the~$k$-dimensional Weisfeiler-Lehman algorithm repeatedly recolors~$k$-tuples of vertices by assigning them a color that depends on the multiset of previous colors of adjacent~$k$-tuples, where tuples are adjacent if they differ by at most one entry. Our restricted version of the algorithm performs this recoloring operation only on a restricted set of~$k$-tuples. For more information on the standard Weisfeiler-Lehman algorithm we refer the reader to existing literature (see~\cite{DBLP:conf/esa/BerkholzBG13} and~\cite{SchweitzerThesis} for more pointers).

For~$k\geq 2$ we now define the \emph{restricted~$k$-dimensional Weisfeiler-Lehman color refinement}. We say a family of sets~$\mathcal{V}$ has \emph{width~$k$} if the largest set in~$\mathcal{V}$ has size~$k$.
Let~$G$ be a graph and~$\mathcal{V}$ be a family of sets of vertices of~$G$ of width at most~$k'$. Let~$\mathcal{V}^+$ be the set of~$k$-tuples~$(v_1,\ldots,v_{k})$ (with entries not necessarily distinct) for which~$\{v_1,\ldots,v_{k'}\}$ is in~$\mathcal{V}$. For every $k$-tuple $(v_1,\ldots,v_k)$ in~$\mathcal{V}^+$ we define $\wl^k_0[\mathcal{V},G](v_1,\ldots,v_k)$ as the isomorphism type of the subgraph induced by the ordered tuple $(v_1,\ldots,v_k)$. 
If the graph is colored then the isomorphism type has to take the coloring into account. More precisely, the coloring~$\wl^k_0$ is a coloring that satisfies~$\wl^k_0[\mathcal{V},G](v_1,\ldots,v_k) = \wl^k_0[\mathcal{V},G](v'_1,\ldots,v'_k)$ if and only if we can map~$v_i$ to~$v'_i$ and obtain an isomorphism of the colored graphs induced by~$\{v_1,\ldots,v_k\}$ and~$\{v'_1,\ldots,v'_k\}$.
If~$(v_1,\ldots,v_{k}) \notin \mathcal{V}^+$ then we define $\wl^k_0[\mathcal{V},G](v_1,\ldots,v_k)$  to be the empty set~$\emptyset$.

Iteratively for $i\geq 0$, we define $\wl^k_{i+1}[\mathcal{V},G](v_1,\ldots,v_k)$ to be the empty set~$\emptyset$ if~$(v_1,\ldots,v_{k'}) \notin \mathcal{V}^+$ and to be~$\big(\wl^k_i[\mathcal{V},G](v_1,\ldots,v_k), {\mathcal{M}}^k_i\big)$ otherwise, where ${\mathcal{M}}^k_i$ is the multiset given by
\[ {\mathcal{M}}^k_i:=  \big\{\!\!\big\{(\wl^k_i[\mathcal{V},G](x,v_2,\ldots,v_k) , \wl^k_i[\mathcal{V},G](v_1,x,v_3,\ldots,v_k),
\ldots,\]\[ \wl^k_i[\mathcal{V},G](v_1,\ldots,v_{k-1},x)) \mid x\in V(G)     \big\}\!\!\big\}.\]
The process partitions the ordered~$k$-tuples into classes according to their color. Since in each iteration the color of the previous iteration is encoded in the new color,~$k$-tuples which are assigned different colors will continue to have different colors in all subsequent iterations. Therefore the refinement process stabilizes. We define $\wl^k_\infty[\mathcal{V},G](v_1,v_2,\ldots,v_k)$ as $\wl^k_i[\mathcal{V},G](v_1,v_2,\ldots,v_k)$ where $i$ is the least positive integer such that the induced partition in step $i$ is equivalent to the induced partition in step~$i+1$.
Abusing notation, we may drop the specifications~$[\mathcal{V},G]$ whenever they are apparent from the context.

\hereorinappendixstatement{
\ifhasappendix{The following lemma is required for the proof of Lemma~\ref{lem:cmp:time:of:res:wl}}\fi
\begin{lemma}\label{lem:props:of:wl}
Suppose for~$i\in \{1,2\}$ we are given a graph $G_i$ and a family of subsets of the vertices~${\mathcal{V}}_i$ of width~$k'
\leq k$. For the restricted~$k$-dimensional Weisfeiler-Lehman color refinement the following properties hold:
\begin{enumerate}
\item \label{first:item:of:lem:props:wl} If~$\wl^k_\infty[\mathcal{V}_1,G_1](v_1,\ldots, v_k) = \wl^k_\infty[\mathcal{V}_2,G_2](v'_1,\ldots,v'_{k})\neq \emptyset$ holds then also for all~$j_1,j_2\in \{1,\ldots,k\}$ it holds that 

\ifllncs
\[\begin{array}{lc} \wl^k_\infty[\mathcal{V}_1,G_1](v_1,\ldots,v_{j_1-1},v_{j_2},v_{j_1+1},\ldots, v_k) &= \\ \wl^k_\infty[\mathcal{V}_2,G_2](v'_1,\ldots,v'_{j_1-1},v'_{j_2},v'_{j_1+1},\ldots,v'_{k}).\end{array}\]
\else
\[ \wl^k_\infty[\mathcal{V}_1,G_1](v_1,\ldots,v_{j_1-1},v_{j_2},v_{j_1+1},\ldots, v_k) \!=\! \wl^k_\infty[\mathcal{V}_2,G_2](v'_1,\ldots,v'_{j_1-1},v'_{j_2},v'_{j_1+1},\ldots,v'_{k}).\]
\fi
\item \label{second:item:of:lem:props:wl} If~$\wl^k_\infty[\mathcal{V}_1,G_1](v_1,\ldots, v_k) = \wl^k_\infty[\mathcal{V}_2,G_2](v'_1,\ldots,v'_{k})\neq \emptyset$ then for all indices~$j_1,j_2 > k'$, if~$v_{j_1}$ and~$v_{j_2}$ are contained in the same connected component of~$G_1 - \{v_1,\ldots,v_{k'}\}$ then~$v'_{j_1}$ and~$v'_{j_2}$ are contained in the same connected component of~$G_2 - \{v'_1,\ldots,v'_{k'}\}$.
\end{enumerate}

\end{lemma}
\begin{proof}
(Part \ref{first:item:of:lem:props:wl}.) If~$\wl^k_\infty[\mathcal{V}_1,G_1](v_1,\ldots, v_k) = \wl^k_\infty[\mathcal{V}_2,G_2](v'_1,\ldots,v'_{k})\neq \emptyset$ then by definition of the stable refinement, 
there exists a vertex~$x\in V_2$ such that the equation 
\begin{center}
\arraycolsep=0.1pt
\ifllncs
$\begin{array}{llcrrllcrr}
\big(\wl^k_{\infty}[\mathcal{V}_1,G_1]& ( & v_{j_2} &,v_2,\ldots,v_k),&  \ldots, &\wl^k_{\infty}[\mathcal{V}_1,G_1]&(v_1,\ldots,v_{k-1}, & v_{j_2} & )\big)& =\\ 
\big(\wl^k_{\infty}[\mathcal{V}_2,G_2]&  ( & x &,v'_2,\ldots,v'_k),&  \ldots,& \wl^k_{\infty}[\mathcal{V}_2,G_2]&(v'_1,\ldots,v'_{k-1}, & x & )\big)
\end{array}$
\else
$\begin{array}{llcrllcrrllcrr}
\big(\wl^k_{\infty}[\mathcal{V}_1,G_1]& ( & v_{j_2} &,v_2,\ldots,v_k),& \wl^k_{\infty}[\mathcal{V}_1,G_1]&(v_1, & v_{j_2} & ,v_3,\ldots,v_k),& \ldots, &\wl^k_{\infty}[\mathcal{V}_1,G_1]&(v_1,\ldots,v_{k-1}, & v_{j_2} & )\big)& =\\ 
\big(\wl^k_{\infty}[\mathcal{V}_2,G_2]&  ( & x &,v'_2,\ldots,v'_k),& \wl^k_{\infty}[\mathcal{V}_2,G_2]&      (v'_1, & x & ,v'_3,\ldots,v'_k),& \ldots,& \wl^k_{\infty}[\mathcal{V}_2,G_2]&(v'_1,\ldots,v'_{k-1}, & x & )\big)
\end{array}$
\fi
\end{center}
holds. Since~$\wl^k_0$ in particular encodes the isomorphism type of the graph induced by its entries, we conclude that~$x = v'_{j_2}$, which proves the first part of the lemma.
\medskip \\
(Part \ref{second:item:of:lem:props:wl}.) We show the following statement by induction on~$t$. If there is a path from~$v_{j_1}$ to~$v_{j_2}$ in~$G_1 - \{v_1,\ldots,v_{k'}\}$ of length at most~$t$, but there is no path of length at most~$t$ from~$v'_{j_1}$ to~$v'_{j_2}$ in~$G_2 - \{v'_1,\ldots,v'_{k'}\}$, then~$\wl^k_\infty[\mathcal{V}_1,G_1](v_1,\ldots, v_k) \neq \wl^k_\infty[\mathcal{V}_2,G_2](v'_1,\ldots,v'_{k})$.
 For~$t= 0$, if there is said path of length~$0$ from~$v_{j_1}$ to~$v_{j_2}$ this implies that~$v_{j_1} = v_{j_2}$ and that~$v_{j_1}\notin \{v_1,\ldots,v_{k'}\}$.

 If there is no path of length~0 from~$v'_{j_1}$ to~$v'_{j_2}$ in~$G_2 - \{v'_1,\ldots,v'_{k'}\}$ then~$v'_{j_1}\neq v'_{j_2}$ or~$v'_{j_1} \in \{v'_1,\ldots,v'_{k'}\}$. Either way, we have that~$\wl^k_\infty[\mathcal{V}_1,G_1](v_1,\ldots, v_k) \neq \wl^k_\infty[\mathcal{V}_2,G_2](v'_1,\ldots,v'_{k})$ since \ifllncs \else the initial coloring\fi~$\wl^k_0$ encodes the isomorphism type of its entries.
 
 Suppose now the statement has been shown for length~$t'< t$ and suppose a shortest path from~$v_{j_1}$ to~$v_{j_2}$ in~$G_1 - \{v_1,\ldots,v_{k'}\}$ is of length~$t>0$. If~$v'_{j_1} \in \{v'_1,\ldots,v'_{k'}\}$ or~$v'_{j_1} \in \{v'_1,\ldots,v'_{k'}\}$ the statement follows as in the induction base since~$\wl^k_0$ encodes the isomorphism type of its entries.
 Also note that~$t>0$ implies~$v_{j_1}\neq v_{j_2}$.
Let~$x$ be a vertex in~$G_1 - \{v_1,\ldots,v_{k'}\}$ which is adjacent to~$v_{j_2}$ and of distance~$t-1$ from~$v_{j_1}$. For every vertex~$y$, either~$v'_{j_2}$ is not adjacent to~$y$ or there is no path of length at most~$t-1$ from~$v'_{j_1}$ to~$y$ in~$G_2 - \{v'_1,\ldots,v'_{k'}\}$. Thus, by induction, we conclude that it cannot be simultaneously the case that
\ifllncs
\[\begin{array}{l}\wl^k_\infty[\mathcal{V}_1,G_1](v_1,\ldots,v_{j_1 -1}, x,v_{j_1 +1} , v_k) =\\ \wl^k_\infty[\mathcal{V}_2,G_2](v'_1,\ldots,v'_{j_1 -1}, y,v'_{j_1 +1} , v'_k)\end{array}\] and also that \[\begin{array}{l}\wl^k_\infty[\mathcal{V}_1,G_1](v_1,\ldots,v_{j_2 -1}, x,v_{j_2 +1} , v_k) =\\\wl^k_\infty[\mathcal{V}_2,G_2](v'_1,\ldots,v'_{j_2 -1}, y,v'_{j_2 +1} , v'_k).\end{array}\] 
\else
$\wl^k_\infty[\mathcal{V}_1,G_1](v_1,\ldots,v_{j_1 -1}, x,v_{j_1 +1} , v_k) =\wl^k_\infty[\mathcal{V}_2,G_2](v'_1,\ldots,v'_{j_1 -1}, y,v'_{j_1 +1} , v'_k)$ and also that~$\wl^k_\infty[\mathcal{V}_1,G_1](v_1,\ldots,v_{j_2 -1}, x,v_{j_2 +1} , v_k) =\wl^k_\infty[\mathcal{V}_2,G_2](v'_1,\ldots,v'_{j_2 -1}, y,v'_{j_2 +1} , v'_k)$. 
\fi
This shows that~$\wl^k_\infty[\mathcal{V}_1,G_1](v_1,\ldots, v_k) \neq \wl^k_\infty[\mathcal{V}_2,G_2](v'_1,\ldots,v'_{k})$. 
 \end{proof}
}

\begin{lemma}\label{lem:cmp:time:of:res:wl}
For a graph~$G$ and a family~$\mathcal{V}$ sets of vertices of~$G$ of width~$k'$, the stable partition of the restricted~$(k'+c)$-dimensional Weisfeiler-Lehman color refinement can be computed in time~$\BigO\big( (k'+c)^2 \cdot   |{\mathcal{V}^+}|  n \cdot \log (|{\mathcal{V}^+}|) \big)$.
\end{lemma}
\proofatend
The technique in this proof, computing the stable partition in said running time, is an adaptation of the technique of Immerman and Lander~\appcite{compleretro} (see also~\cite{DBLP:conf/esa/BerkholzBG13}) that computes the stable refinement of the (unrestricted) Weisfeiler-Lehman algorithm in time~$\BigO(k^2 n^{k+1} \log(n))$. We define~$k = k'+c$. If the given graph is uncolored, the initial coloring of~$\wl_0^k(v_1,\ldots,v_k)$ can be chosen to be the string~$E(v_1,v_1),\ldots, E(v_1,v_k),E(v_2,v_1),\ldots, E(v_k,v_k)$, where~$E(v_i,v_j)= 1$ if~$(v_i,v_j)$ is an edge and 0 otherwise. If the graph is already equipped with a coloring, 
we define an initial coloring as the pair of the color just defined and the initially given color.
Given  the initial coloring we proceed as follows: We again let~$\mathcal{V}^+$ be the set of all~$k$-tuples~$(v_1,\ldots,v_{k})$ for which~$\{v_1,\ldots,v_{k'}\}$ is in~$\mathcal{V}$.
We use a list~$S$ that will contain subsets of~$\mathcal{V}^+$. Initially it contains the color classes according to the initial coloring.
The algorithm now repeatedly splits color classes using the first element of the list, which is removed afterwards. More specifically, suppose~$B$ is the first element of the list. Let~$B'$ be the elements of~$\mathcal{V}^+$ which can be obtained by replacing up to one vertex in an element of~$B$. 
For each element~$b'$ of~$B'$ we compute a new color using a modified Weisfeiler-Lehman recursion. That is, if~$\chi$ is the previous coloring,
 the new color of~$b'= (b_1,\ldots,b_{k})$ will be~$\chi'(b')= 
\big(\chi(b'), {\mathcal{M}}^k_i\big)$, where ${\mathcal{M}}^k_i$ is the multiset given by
\[\big\{\!\!\big\{(\chi(x,b_2,\ldots,b_k) , 
\ldots, \chi(b_1,\ldots,b_{k-1},x)) \mid x\in V(G)      \text { and there  }\]\[\text{ exists } i \in\{1,\ldots,k\}  \text{ such that } (b_1,\ldots,b_{i-1},x,b_{i+1},\dots,b_k) \in B \big\}\!\!\big\}.\]

For each color class that is split into several smaller color classes during the procedure, we add all but the largest of these smaller color classes to the end of the list~$S$.
The algorithm continues until the list is empty. The stable partition is the partition induced by the final color classes.
To prevent the names of colors from becoming excessively long strings, we always rename newly arising colors by assigning previously unused integers. Details are given in Algorithm~\ref{algo:restricted:wl}.

\emph{Correctness.} The computed partition is coarser than or equal to the stable partition of the restricted Weisfeiler-Lehman algorithm, since each step splits color classes according to color induced subsets of the multisets from in the definition of the restricted Weisfeiler-Lehman algorithm using previously computed colors. For these previously computed colors, the partition is coarser than or equal to the stable partition by induction.
For correctness it thus suffices to argue that when given the final coloring, which we denote by~$\chi$, the restricted Weisfeiler-Lehman algorithm does not refine the color classes.

Supposing otherwise, there must be 
tuples~$(b_1,\ldots,b_k)$ and~$(b'_1,\ldots,b'_k)$ in~$\mathcal{V}^+$, such that the restricted Weisfeiler-Lehman algorithm assigns these tuples different colors in the first iteration when provided with~$\chi$ as initial colors.
Thus, there must be a~$k$-tuple of color classes~$(B_1,\ldots,B_k)$ such that the sets~$\{x\in V(G)\mid \chi(x,b_2,\ldots,b_k) \in B_1,\ldots,\chi(b_1,\ldots,b_{k-1},x)\in B_k \}$ and~$\{x\in V(G)\mid \chi(x,b'_2,\ldots,b'_k) \in B_1,\ldots,\chi(b'_1,\ldots,b'_{k-1},x)\in B_k \}$ have different cardinality.
Let~$B_i$ be the color class among~$\{B_1,\ldots,B_k\}$ for which at the latest time during the algorithm there is some superset~$\hat{B}$ in the list~$S$. If this set is~$B_i$ then Algorithm~\ref{algo:restricted:wl} would have distinguished~$(b_1,\ldots,b_k)$ and~$(b'_1,\ldots,b'_k)$. This shows that~$B_i$ was the largest class among the classes into which the superclass~$\hat{B}$ was split. However, in this case, among the classes into which~$\hat{B}$ was split, there is a class that can replace~$B_k$ in the sequence~$(B_1,\ldots,B_k)$ yielding a new sequence for which we also obtain two sets of different cardinality.

\emph{Running time.} To bound the running time, it suffices to bound the amount of work performed in the for-loop of Algorithm~\ref{algo:restricted:wl}. Note that 
whenever a set is added to the list, its size is at most half of the size of the class that is split. Thus every element can be in at most~$\BigO(\log(|{\mathcal{V}}^+|))$ 
 many of the sets which are at some point in~$S$. It suffices now to show that it is possible to compute~$\chi'(b')$ for all elements in~$B'$ in time~$\BigO((k'+c)^2 \cdot   |B'|n)$. Note that every element in the multiset used to compute~$\chi'(b')$ contains at least one coordinate that is in~$B$. We can therefore compute~$\chi'(b')$ for all elements in~$B'$ by determining the counts of elements in said multiset as follows. We iterate over all elements~$b$ of~$B$ and all coordinate positions~$j$ in~$\{1,\ldots,k\}$ and add suitable counts to all elements~$b'$ in~$B'$ for which we can obtain~$b$ by replacing coordinate~$j$ of~$b'$ with coordinate~$j$ of~$b$. The~$k$-tuple of colors which is to be counted is obtained by inserting this value of coordinate~$j$ of tuple~$b$ one by one into the other coordinates of~$b'$. To avoid double counting, the corresponding counter is then incremented if~$j$ is the least position for which the entry is in~$B$. Iterating over the elements in each color class that is split, we can in each case find the largest color class. In time linear in the sizes of the split color classes.
\begin{algorithm}[t]
\caption{The Weisfeiler-Lehman algorithm restricted to~$\mathcal{V}$}
\label{algo:restricted:wl}
\begin{algorithmic}[1]
\REQUIRE Integers~$k'$ and~$c$, a graph~$G$, a family~$\mathcal{V}$ of sets of vertices of~$G$ of width~$k'$, and a possibly uniform coloring~$\chi$ of all~$k$-tuples in~$\mathcal{V}^{+}$ (where~$k=  (k'+c)$).
\ENSURE  The stable partition of the restricted~$(k'+c)$-dimensional Weisfeiler-Lehman color refinement. 
\ENSUREGAP
\STATE for all tuples~$(v_1,\ldots,v_{k})\in \mathcal{V}^+$   define  the color $\wl^k_0(v_1,\ldots,v_{k})$ as $(\chi(v_1,\ldots,v_{k}), (E(v_1,v_{k}),E(v_2,v_1),\ldots, E(v_k,v_{k}))) $; for all other tuples the initial coloring is~$\emptyset$
\STATE initialize an empty list~$S$
\STATE add each color class in~$\mathcal{V}^+$ to the list in lexicographic order
\WHILE {$S$ is not empty}
	\STATE let~$B$ be the first set of~$S$
	\STATE let~$B'$ be the set of elements~$b'\in \mathcal{V}^+$ for which there exists a set~$b\in B$ differing from~$b'$ by at most one vertex
	\STATE compute~$\chi'(b')$ for all~$b'\in B'$
	\STATE replace~$\chi$ with~$\chi'$ for elements in~$B'$
	\STATE for all color classes that are split due to this replacement of~$\chi$ add all but the largest of the new color classes to the end of the list~$S$
	
\STATE remove~$B$ from the front of the list
\ENDWHILE
\RETURN the partition induced by~$\chi$
\end{algorithmic}
\end{algorithm}
\endproofatend

\hereorinappendixstatement{
We remark that when the described efficient version of the restricted Weisfeiler-Lehman algorithm is used on two graphs for comparing colors of tuples between the two graphs, some care has to be taken concerning hashing. If hashing of colors is used, then the algorithm has to be performed on both graphs at the same time (or alternatively on the disjoint union of the two graphs) to avoid hashing different color to the same value.
}

Being a restriction implies that the known examples that cannot be solved by the Weisfeiler-Lehman algorithm~\cite{cfi:paper} can also not be solved by the restricted version. However, we will now prove that the restricted Weisfeiler-Lehman algorithm decides isomorphism of graphs whenever the set~$\mathcal{V}(G)$ captures a tree decomposition. To facilitate the proof and to make it more easily applicable in the future, we prove the theorem for tree decompositions, instead of proving it just for strong tree decompositions.

Recall that a \emph{tree decomposition} is a pair $(\{X_{i} \mid i \in I\}, T = (I,F))$ for which~$\bigcup_{ i \in I} X_{i} = V(G)$ and~$T=(I,F)$ is a tree such that every vertex is contained in some bag, for adjacent vertices there is a bag containing both of them and for every vertex~$v$ the set of bags containing~$v$ induces a connected subtree of~$T$.

Given a graph~$G$ we say that a family of sets~$\mathcal{V}(G)$ \emph{captures a tree decomposition~$\mathcal{T}$} of~$G$ if 
every bag is in~$\mathcal{V}(G)$. If~$G$ is equipped with an equivalence relation and possibly a tuple-coloring, we additionally require that every equivalence class is contained in a bag of~$\mathcal{T}$. 
We say that a tree decomposition is \emph{semi-smooth} if the intersection of adjacent bags has size at most one smaller that the size of the larger bag (for the decomposition to be smooth one also requires that all bags have the same size, see~\cite{DBLP:journals/siamcomp/Bodlaender96}). 

\begin{theorem}\label{thm:restricted:wl:solves:treewidth}
Suppose we are given an algorithm that computes for every graph~$G$ in a graph class~$\mathcal{C}$ an isomorphism invariant family of vertex sets~$\mathcal{V}(G)$ of width at most~$k'$ such that~$\mathcal{V}(G)$ captures a semi-smooth tree decomposition of~$G$.
Then we can decide isomorphism of graphs in~$\mathcal{C}$ with the~$(k'+3)$-dimensional Weisfeiler-Lehman algorithm restricted to~$\mathcal{V}(G)$.
\end{theorem}
\proofatend

Suppose~$G$ and~$G'$ are graphs in~$\mathcal{C}$. We run the Weisfeiler-Lehman algorithm and claim the graphs are isomorphic if and only if the multisets of colors obtained as the stable refinement of each graph are isomorphic.

Since functor~$\mathcal{V}(\cdot)$ is isomorphism invariant, the restricted Weisfeiler-Lehman algorithm is also isomorphism invariant. Therefore, if~$G$ and~$G'$ are isomorphic then, this procedures claims that the graphs are isomorphic.

We now show that if the restricted Weisfeiler-Lehman algorithm terminates with the multisets of colored tuples being equal in both graphs, then the graphs are isomorphic. 

Let~$\mathcal{T}$ be a semi-smooth tree decomposition of~$G$ that is captured by~$\mathcal{V}(G)$. Let~$k = k'+3$.  
We now prove the following statement by induction on~$t$: 

\begin{claim} If~$\wl^k_\infty(v_1,v_2,\ldots,v_{k'},z,z,z) = \wl^k_\infty(v'_1,v'_2,\ldots,v'_{k'},z',z',z')$ and if the set~$\{v_1,\ldots,v_{k'}\}$ is a bag of~$\mathcal{T}$,
then there is an isomorphism~$\phi$ from the subgraph of~$G$ induced by the vertices of~$\{v_1,\ldots,v_{k'}\}$ and all vertices in components of~$G\setminus  \{v_1,\ldots,v_{k'}\}$ not containing~$z$ of size at most~$t$ to the subgraph of~$G'$ induced by the vertices of~$\{v'_1,\ldots,v'_{k'}\}$ and all vertices in components of~$G'\setminus  \{v'_1,\ldots,v'_{k'}\}$ not containing~$z'$ of size at most~$t$. Moreover the isomorphism~$\phi$ can be chosen such that~$v_i$ maps to~$v'_i$ for~$i\in \{1,\ldots,k\}$ and such that it respects colors of equivalence classes.
\end{claim}

Under the assumption, the multiset~$\wl^k_\infty(v_1,v_2,\ldots,v_{k'},z,z,z)$ is non-empty, which implies that the set~$\{v'_1,v'_2,\ldots,v'_{k'}\}$ is contained in an element of~$\mathcal{V}(G')$.
\ifllncs\else

\fi
For~$t = 0$, \
by definition~$\wl^k_0(v_1,v_2,\ldots,v_{k'},z,z,z) = \wl^k_0(v'_1,v'_2,\ldots,v'_{k'},z',z',z')$ if and only if the map sending~$v_i$ to~$v'_i$ and~$z$ to~$z'$ is an isomorphism from the graph~$G[\{v_1,\ldots,v_{k}, z\}]$ to the graph~$G'[\{v'_1,\ldots,v'_{k}, z'\}]$. This resolves the base case.

Suppose now~$t>0$ and \ifllncs also suppose \fi that~$\wl^k_\infty(v_1,v_2,\ldots,v_{k'},z,z,z) = \wl^k_\infty(v'_1,v'_2,\ldots,v'_{k'},z',z',z')$ where the set~$\{v_1,\ldots,v_{k'}\}$ is a bag of~$\mathcal{T}$.

Let~$B=  \{v_1,\ldots,v_{k'}\}$and~$B'=  \{v'_1,\ldots,v'_{k'}\}$. We call the components of~$G -B$ and~$G'-B'$ that do not contain~$z$ or~$z'$, respectively, and which are of size at most~$t$ the \emph{components of interest}.
Let~$C$ be a component of interest in~$G - B$. Since~$\mathcal{T}$ is semi-smooth, there is a bag~$\hat{B}$ in~$\mathcal{T}$ containing a vertex~$w$ of~$C$ and all but one of the vertices from~$B$. Moreover there is a vertex in~$\{v_1,\ldots,v_{k'}\}$ such that no vertex of~$C$ different from~$w$ is adjacent to this vertex. Without loss of generality we assume that~$v_1$ fulfills this property.

Since~$\wl^k_\infty(v_1,\ldots,v_{k'},z,z,z) = \wl^k_\infty(v'_1,\ldots,v'_{k'},z',z',z')$ there must be a vertex~$w'$ in~$G'$ such that~$\wl^k_\infty(v_1,v_2,\ldots,v_{k'},z,w,z) = \wl^k_\infty(v'_1,v'_2,\ldots,v'_{k'},z',w',z')$ holds. This implies that~$w' \notin \{v'_1,\ldots,v'_{k'}\}$. By Lemma~\ref{lem:props:of:wl}, this also implies that~$w'$ is not in the same component as~$z'$ in~$G'- \{v'_1,\ldots,v'_{k'}\}$. Let~$C'$ be the component of~$G' - \{v'_1,\ldots,v'_{k'}\}$ containing~$w'$. We show that the map that sends~$v_i$ to~$v'_i$ can be extended to an isomorphism that maps~$C$ to~$C'$.
Since by Lemma~\ref{lem:props:of:wl} the equation~$\wl^k_\infty(v_1,v_2,\ldots,v_{k'},z,w,z) = \wl^k_\infty(v'_1,v'_2,\ldots,v'_{k'},z',w',z')$ in particular implies~$\wl^k_\infty(w,v_2,\ldots,v_{k'},v_1,v_1,v_1) = \wl^k_\infty(w',v'_2,\ldots,v'_{k'},v'_1,v'_1,v'_1)$, the map sending~$w$ to~$w'$ and~$v_i$ to~$v'_i$ for~$i\geq 2$ extends to an isomorphism mapping all components of~$G-\{w,v_2\ldots,v_{k'}\}$ of size at most~$t-1$ not containing~$v_1$ to components of~$G-\{w',v'_2\ldots,v'_{k'}\}$ not containing~$v'_1$. Since~$C$ and~$C'$ are connected and contain~$w$ and~$w'$ respectively, this implies that this mapping maps~$C$ to~$C'$.

We summarize what we have achieved so far towards proving the claim: for every component~$C$ of~$G-B$ of interest (i.e., of size at most~$t$ and not containing~$z$) we have found a component~$C'$ of~$G'-B'$ of interest (i.e., of size at most~$t$ and not containing~$z'$) such that~$G[C \cup B]$ can map isomorphically to~$G'[C\cup B']$. We now need to show that we can find an isomorphism that simultaneously maps all components of interest in~$G-B$ bijectively to the components of interest in~$G'-B'$.

From what we have shown so far, it could in principle be possible that by choosing a different vertex of~$C$ as~$w$, the new corresponding vertex~$w'$ would have to be taken outside of~$C'$ and thus give us a different component~$C''$ of~$G'- B'$. To argue that this is not the case, we show that~$\{\!\{\wl^k_\infty(v_1,\ldots,v_{k'},z,u,z)\mid u\in C\}\!\} = \{\!\{\wl^k_\infty(v'_1,\ldots,v'_{k'},z',u',z')\mid u'\in C'\}\!\}$ as multisets. 
To show this, it suffices by Lemma~\ref{lem:props:of:wl} to show that~$\{\!\{\wl^k_\infty(v_1,v_2,\ldots,v_{k'},z,u,w)\mid u\in C\}\!\} = \{\!\{\wl^k_\infty(v'_1,v'_2,\ldots,v'_{k'},z',u',w')\mid u'\in C'\}\!\}$.

Since~$\wl^k_\infty(v_1,\ldots,v_{k'},z,w,z) = \wl^k_\infty(v'_1,\ldots,v'_{k'},z',w',z')$ it follows also from Lemma~\ref{lem:props:of:wl} that~$\wl^k_\infty(w,v_2,\ldots,v_{k'},z,z,w) = \wl^k_\infty(w',v'_2,\ldots,v'_{k'},z',z',w')$ and consequently it follows that as multisets~$\{\!\{\wl^k_\infty(v_1,v_2,\ldots,v_{k'},z,u,w)\mid u\in V\}\!\} = \{\!\{\wl^k_\infty(v'_1,v'_2,\ldots,v'_{k'},z',u',w')\mid u'\in V\}\!\}$.
It suffices now for us to show that if~($x\in C$ and~$x'\notin C'$) or if~($x\notin C$ and~$x'\in C'$) then automatically also the inequality~$\wl^k_\infty(v_1,v_2,\ldots,v_{k'},z,x,w)\neq \wl^k_\infty(v'_1,v'_2,\ldots,v'_{k'},z',x',w')$ holds. This however follows again by Lemma~\ref{lem:props:of:wl}.

We define the sought isomorphism\ifllncs as follows\fi. Since~$\wl^k_\infty(v_1,\ldots,v_{k'},z,z,z) = \wl^k_\infty(v'_1,\ldots,v'_{k'},z',z',z')$, we can find isomorphism from the components of~$G-B$ of interest to the components of~$G'-B'$ of interest such these isomorphism are compatible with the map from~$B$ to~$B'$ mapping~$v_i$ to~$v'_i$. This yields an isomorphism from the graph induced by all vertices in~$B$ and all vertices in components of interest bijectively to those of~$G'$. 

To finish the proof of the claim, it remains to argue that the isomorphism we constructed respects the coloring of the graphs. Let~$(v_1,\ldots,v_t)$ be an ordering of an equivalence class of the vertices of~$G$.
Since there is a bag of~$\mathcal{V}(G)$ that contains the equivalence class,\ifllncs \else the multiset\fi~$\wl^k_\infty(v_1,v_2,\ldots,v_{t},v_t,\ldots ,v_t)$ is not empty. Moreover under the isomorphism we have constructed, the color of the image,~$\wl^k_\infty(v'_1,v'_2,\ldots,v'_{t},v'_t,\ldots ,v'_t)$ must be non-empty and thus the induced subgraph~$G[v_1,\ldots,v_t]$ must map to the induced subgraph~$G'[v'_1,\ldots,v'_t]$ respecting the coloring of the equivalence class. By symmetry, colored orderings of equivalence classes of~$G'$ also map to orderings of equivalence classes of~$G$ of the same color. This proves the claim.

We now apply the claim to finish our proof. 
By setting~$z$ equal to~$v_{k'}$ and~$t$ equal to~$n+1$ we obtain from the claim the existence of an isomorphism that maps all vertices of~$G$ bijectively to the vertices of~$G'$.
\endproofatend

The previous theorem requires~$\mathcal{V}(G)$ to capture a semi-smooth tree decomposition. However, we can extend the theorem to tree decompositions and strong tree decompositions by using the alternative set~$\mathcal{V'} = \{B_1 \cup B_2 \mid B_1,B_2 \in \mathcal{V}(G)\}$. This can be seen by the following two observations.

If~$\mathcal{V}$ is isomorphism invariant and captures a tree decomposition of~$G$, then~$\mathcal{V'}$ is isomorphism invariant and captures a semi-smooth tree decomposition. This follows from the construction that produces a smooth tree decomposition from a tree decomposition given in~\cite{DBLP:journals/siamcomp/Bodlaender96}.

Suppose~$\mathcal{B}$ is the set of bags of a strong tree decomposition. Then there is a semi-smooth tree decomposition~$\mathcal{B'}$ such that for every bag~$B$ of~$\mathcal{B'}$ there are bags~$B_1$ and~$B_2$ in~$\mathcal{B}$ such that~$B\subseteq B_1 \cup B_2$. This can be seen with the standard way of constructing  a tree decomposition from a strong tree decomposition by inserting for each edge between two bags~$B_1$ and~$B_2$ a path of bags transforming~$B_1$ to~$B_2$ by replacing successively one vertex after the other. We conclude, if~$\mathcal{V}$ captures a strong tree decomposition then~$\mathcal{V'}$ captures a semi-smooth tree decomposition.
This also shows that by setting~$\mathcal{V}$ to be the set of all~$k$-tuples of vertices, every graph of tree width at most~$k$ has a smooth tree decomposition captured by~$\mathcal{V'}$ and shows that the a sufficiently high-dimensional Weisfeiler-Lehman algorithm solves graph isomorphism of graphs of bounded tree width, as mentioned in the introduction.

\ifhasappendix\else Concerning the previous two observations, in our context relevant is the following fact. If we start with a family~$\mathcal{V}$ of fpt size in some parameter~$k$ with the maximum set size bounded by a function of~$k$ and we apply either construction, we obtain a family~$\mathcal{V'}$ that is also of fpt size.\fi

\begin{corollary}\label{cor:restricted:wl:solves:treewidth:no:smooth}
For a parameter~$k'$, given an fpt-algorithm that computes for every graph~$G$ in a graph class~$\mathcal{C}$ an isomorphism invariant family of vertex sets~$\mathcal{V}(G)$ of width at most~$k'$ such that~$\mathcal{V}(G)$ captures a tree decomposition (or a strong tree decomposition), isomorphism of graphs in~$\mathcal{C}$ is fixed parameter tractable in~$k'$.
\end{corollary}
\proofatend
Since the algorithm computing~$\mathcal{V}$ is an fpt-algorithm, the size of~$\mathcal{V}$ is bounded by~$f(k') n^d$ for some function~$f$ and constant~$d\in \mathbb{N}$.
The set~$\mathcal{V'}$ as defined above has thus size at most~$f(k')^2 n^{2d}$, width at most~$2k'$, and captures a smooth tree decomposition. 
Theorem~\ref{thm:restricted:wl:solves:treewidth} shows that it suffices to perform the~$(2k'+3)$-dimensional Weisfeiler-Lehman algorithm restricted to~$\mathcal{V'}$ to solve isomorphism, which can be done im fpt time by Lemma~\ref{lem:cmp:time:of:res:wl}.
\endproofatend

We remark that if we were interested in actual running times, in the tree width case it is possible to avoid that the increase in width from~$\mathcal{V}$ to~$\mathcal{V'}$ yielding better running times bounds.

\section{Tree distance decompositions with connected root bags}\label{sec:distance:decomp}
\differingappendixstatement{}{\section{Proofs of Section~\ref{sec:distance:decomp}}}

As a first application we show that graph isomorphism parameterized by root connected tree distance width is fixed parameter tractable. 
We say two vertices~$v_1$ and~$v_2$ in a graph are~\emph{$k$-connected},  if there are~$k$ internally vertex disjoint paths from~$v_1$ to~$v_2$. We denote this by~$v_1\kcon_k v_2$. The task of checking whether two vertices are~$k$-connected is also known as the Menger Problem and can be solved in polynomial time via a reduction to the maximum flow problem. 

\begin{lemma}\label{lem:must:be:in:same:bag}
Let~$G$ be a graph containing vertices~$v_1$ and~$v_2$.
If~$v_1\kcon_{2k} v_2$, then in every strong tree decomposition of width at most~$k$ the vertices~$v_1$ and~$v_2$ are in the same bag.
\end{lemma}
\proofatend
Suppose~$v_1$ and~$v_2$ are in different bags~$B_v$ and~$B_{v'}$ of a strong tree decomposition of~$G$ of width at most~$k$. Let~$B'$ be the bag adjacent to~$B_v$ in the tree along the path from~$B_{v}$ to~$B_{v'}$. (It is possible that~$B_{v}= B'$.) Suppose there are~$2 k$ internally vertex-disjoint paths from~$v_1$ to~$v_2$. Each of these paths must contain an edge with an endpoint in~$B_{v}$ and an endpoint in~$B'$. Let~$M$ be a set of~$2k$ such edges each belonging to a different path. For every two edges in~$M$, 
the end vertices must be distinct unless they are~$v_1$ or~$v_2$.
Since the bag~$B_v$ has size at most~$k$ there are at most~$k$ edges in~$M$ that have~$v_1$ as an endpoint. Moreover there are at most~$k-1$ edges in~$M$ that have an endpoint in~$B_v$ different from~$v_1$. This shows that~$M$ has size at most~$2k-1$, contradicting the existence of~$2k$ internally vertex-disjoint paths between~$v_1$ and~$v_2$.
\endproofatend

A similar result for tree decompositions, stipulating the existence of a bag that contains both~$v_1$ and~$v_2$ in tree decompositions of width at most~$k$ whenever~$v_1\kcon_{k+1}v_2$, can be found in~\cite{DBLP:journals/jco/Bodlaender03}. \ifhasappendix{The previous lemma restricts the possible bags in a decomposition and leads to an efficient algorithm for isomorphism parameterized by root-connected tree distance width.}\fi 

\hereorinappendixstatement{
Suppose we are searching for a root set~$S$ and have already found a subset~$S'$ of~$S$. In this case, the previous lemma provides means of finding other vertices that must also be contained in~$S$. The following lemmas are also concerned with extending~$S'$ to a suitable root set.
\ifhasappendix{They are required for the proof of Theorem~\ref{thm:rctdw:fpt}.}\fi

\begin{lemma}\label{lem:need:not:touch:fin:cmp}
Let~$G$ be a connected graph and let~$S$ be a root set such that~$\tdw_S(G)  = k$. Let~$S'$ be a subset of~$S$. Let~$C$ be a component of~$G-S'$.
\begin{enumerate}
\item If~$\tdw_{S'}(G[S'\cup C]) > k$ then~$(S -S') \cap C$ is non-empty.\label{must:touch}
\item If~$\tdw_{S'}(G[S'\cup C]) \leq k$ then~$\tdw_{S\setminus C}(G) \leq k$.
\label{do:not:have:to:touch}
\end{enumerate}
\end{lemma}

\begin{proof} For the first part, if~$(S-S')\cap C$ were empty then the every tree distance decomposition on~$G$ with root set~$S$ induces a tree distance decomposition on~$G[S'\cup C']$ with root set~$S'$, which contradicts~$\tdw_{S'}(G[S'\cup C]) > k$. 

For the second part let~$\mathcal{S}$ be a tree distance  decomposition of~$G$ with root set~$S$ of width at most~$k$. Let~$\mathcal{S'}$ be a tree distance decomposition of~$G[S'\cup C]$ with root set~$S'$ of width at most~$k$.
We construct a new tree distance decomposition of~$G$: The set of bags is the union of the sets~$M_1 = \{B\setminus C \mid B\in \mathcal{S} \}$ and~$M_2 = \mathcal{S'} \setminus \{S'\}$. The root bag of the decomposition is~$S \setminus C$. It suffices now to show that this is a rooted tree distance decomposition of width at most~$k$. Note that, together, the sets~$M_1$ and~$M_2$ partition the vertices of~$G$.
To see that the graph induced by this partitioning is a tree, first note that all neighbors of a vertex contained in a set in~$M_2$ which are contained in a bag in~$M_1$ are vertices in~$S'$. Since~$S\setminus C\supset S'$ is a bag, it suffices thus to show that~$M_1$ and~$M_2 \cup \{S'\}$ form strong tree decompositions on~$G\setminus C$ and~$G[C\cup S']$ respectively. For the latter this follows since it constitutes the original decomposition of~$G[S'\cup C]$. For the former this follows since the graph is connected and the new bags are subsets of distinct former bags, and therefore the formed decomposition cannot have cycles. Since all bags of the new decomposition are subsets of previous bags the new decomposition has width at most~$k$.
\end{proof}

\begin{lemma}\label{lem:no:highly:connected:vertex:pair:implies:bd:degree}
Let~$G$ be a graph and~$S$ be a root set such that~$\tdw_S(G)  = k$. Let~$S'$ be a subset of~$S$. If~$C$ is a component of~$G-S'$ with $|N(S')\cap C| > \ell$ for some positive integer~$\ell\in \mathbb{N}$ then there exists~$v \in C$ which is connected by more than~$(\ell - k)/{k^2}$ internally vertex-disjoint paths to~$S'$.
\end{lemma}
\begin{proof}
If~$\ell \leq k$ there is nothing to show, so we assume otherwise. Consider a strong tree decomposition of width at most~$k$ with root set~$S$ and consider the components~$C_1,\ldots,C_t$ of~$G-S$ that contain a vertex of~$N(S') \cap C$. Let~$C_i$ be such a component. Since every bag contains at most~$k$ vertices, and since vertices in~$N(S) \cap C_i$ must be in the same bag, each component~$C_i$ contains at most~$k$ vertices of~$N(S')\cap C$. Since~$|N(S') \cap C| > k$, this implies that there must be some vertex in~$v \in S\setminus S'$ that has a neighbor in~$C_i$. In particular, we can find a path from some vertex in~$S'$ to~$v$ whose internal vertices all lie in~$C_i$.
Each vertex in~$N(S')\cap C$ must either be contained in~$S$ or in a component~$C_i$ for some~$i\in \{1,\ldots,t\}$. 
Since each~$C_i$ can only contain~$k$ vertices of~$N(S')\cap C$ and since~$|N(S')\cap C| >  \ell$ there must be more than~$(\ell-k) /k$ components. This leads to more than~$(\ell-k) /k$ internally vertex-disjoint paths from vertices in~$S'$ to vertices in~$S\setminus S'$. By the pigeonhole principle there is a vertex~$v\in S\setminus S'$ connected by more than~$((\ell-k) /k)/{k}= (\ell-k) /{k^2}$ internally vertex-disjoint paths to~$S'$.
Note that, since~$C$ is a connected component, the vertex~$v$ is in~$C$.
\end{proof}

The three preceding lemmas can be assembled to find suitable root sets, if they exist, allowing for efficient isomorphism tests.
 }
\begin{theorem}\label{thm:rctdw:fpt}
The isomorphism problem parameterized by root-connected tree distance width can be solved in fpt time.
\end{theorem}
\proofatend 

As explained in the introduction given a set~$S$ we can compute in polynomial time a unique minimal distance decomposition with root set~$S$ (\cite[Theorem 2.1]{DBLP:journals/algorithmica/YamazakiBFT99}). By simply inspecting the size of all bags, we can also check in polynomial time whether this decomposition has width at most~$k$.

By Corollary~\ref{cor:restricted:wl:solves:treewidth:no:smooth} it thus suffices for us to show that we can enumerate in fpt time an isomorphism invariant family of root sets that contains at least one root set that yields a tree distance width of at most~$k$. Indeed for each such root set we can compute the minimal distance decomposition and collect all possible bags of decompositions of width at most~$k$ obtained this way. The family of sets obtained has size at most~$f(k)n^d $ for some function~$f$ and some integer~$d$ and captures a strong tree decomposition.

We now describe an algorithm for the task of enumerating the family of root sets, details are given in Algorithm~\ref{algo}. We remark however, that the suitable root set computed by Algorithm~\ref{algo} is not necessarily connected and the parameter root-connected tree distance width is used only in the analysis of the running time.
\begin{algorithm}[t]
\caption{An fpt algorithm for isomorphism parameterized by root-connected tree distance width~$\isoalgo(G,S')$}
\label{algo}
\begin{algorithmic}[1]
\REQUIRE A graph~$G$ a subset of the vertices~$S'$ and a parameter~$k$.
\ENSURE  A non-empty isomorphism invariant list of sets of~$S$ for which~$S\supset S'$ for which~$\tdw_S(G)\leq k$, or report \textbf{false}, if no such set exists. 
\ENSUREGAP
\IF {$S' = \{\}$}
\FORALL {$v \in V(G)$}
\STATE $S'\leftarrow {v}$
\STATE~$\isoalgo(G,S')$
\ENDFOR
\ELSE
\WHILE {there is a vertex~$v\notin S'$ from which there are~$2k$ internally vertex disjoint paths to~$S'$}
\STATE $S' \leftarrow S' \cup \{v\}$.
\ENDWHILE
\IF  {$\tdw_S(G)\leq k$}
	\STATE print~$S$
\ELSE
	\IF {$|S'| < k$}
		\STATE let~$C_1,\ldots,C_t$ be the components~$C$ of~$G\setminus S'$ for which~$\tdw_{S'}(G[S'\cup C]) > k$
		\IF {$t+|S'|\leq k$ and~$|N(S') \cap C_i| \leq 2k^3+k$  for all~$i\in \{1,\ldots,t\}$}
			\FORALL {$(v_1,\ldots,v_t) \in N(S') \cap C_1 \times \dots \times N(S') \cap C_t$}
			
				\STATE $S'\leftarrow S'\cup \{v_1,\dots,v_t\}$
				\STATE $\isoalgo(G,S')$
			\ENDFOR
		\ENDIF
	\ENDIF
\ENDIF
\ENDIF
\end{algorithmic}
\end{algorithm}
To enumerate root sets, we start with an arbitrary vertex which we add to a previously empty set~$S'$.
While there is a vertex~$v'$ not in~$S'$ from which there are~$2k$ internally vertex disjoint paths to~$S'$ we add~$v'$ to~$S'$.

If~$S'$ induces a tree distance decomposition of width at most~$k$ we output~$S'$. Otherwise, let~$C_1,\ldots,C_t$ be the components~$C$ of~$G\setminus S'$ for which~$\tdw_{S'}(G[S'\cup C]) > k$.
If~$t+|S'|> k$ or~$|N(S') \cap C_i| > 2k^3+k$  for some~$i\in \{1,\ldots,t\}$ we assess~$S'$ as infeasible and backtrack.
In case we do not asses~$S'$ as infeasible, we add from each component~$C_i$ one every vertex contained in~$N(S') \cap C_i$ to~$S'$ and recurse. By branching, we try all possibilities of adding vertices to~$S'$ in this way. That is, for each possible combination of picking one vertex from each~$N(S') \cap C_i$ we recursively continue the process. 

Whenever this process yields a set~$S'$ that induces a tree distance decomposition of width at most~$k$ we declare~$S'$ as a candidate and output it.

Since the vertices that are added to~$S'$ are always added from a set of vertices defined by an isomorphism invariant property, the sets of candidates computed by the algorithm is isomorphism invariant.

\emph{Correctness.} To prove the correctness of the algorithm, it remains to show that this algorithm produces a root set that yields a tree distance width of at most~$k$.
Let~$S$ be a connected set with~$\tdw_S(G) \leq k$. Let~$S''$ be a maximal subset of~$S$ among the sets~$S'$ generated by the algorithm. We show that~$\tdw_{S'}(G) \leq k$.
For contradiction, we assume otherwise.
Since the algorithm initially starts with an arbitrary vertex,~$S''$ is non-empty.
Consider the quotient graph in which the vertices of~$S'$ have been identified. By Lemma~\ref{lem:must:be:in:same:bag}, if there is a vertex~$v$ not in~$S''$ from which there are ${2k}$ internally vertex-disjoint paths to~$S''$ then~$v$ must be in~$S$. The algorithm will add~$v_2$ to~$S''$. Moreover all vertices added  to~$S''$  due their connectivity must be contained in~$S$. Thus the algorithm produces a larger subset contained in~$S$ contradicting maximality of~$S''$. If on the other hand no such vertex~$v_2$ exists, by Lemma~\ref{lem:need:not:touch:fin:cmp} there exists at least one component (Part~\ref{do:not:have:to:touch}), but at most~$k$ (Part~\ref{must:touch}) components~$C$ with~$\tdw_{S''}(G[S''\cup C]) > k$. By Lemma~\ref{lem:no:highly:connected:vertex:pair:implies:bd:degree} for each such component~$|N(S'') \cap C| \leq 2k^3+k$. Since~$S$ is connected, for each such component there is a vertex in~$N(S'') \cap C \cap S$. Thus there is a choice for the algorithm of one vertex for each component that yields a larger subset of~$S$ contradicting maximality of~$S''$.

\emph{Running time.} The running time of Algorithm~\ref{algo} is~$\BigO(k^{2k^3+k} n^{c})$ since each call of the algorithm can be performed in polynomial time, and every vertex, except the first, which is to be added to~$S'$ is chosen from a set of size at most~$2k^3+k$ bounding the number of recursive calls to the algorithm by~$\BigO(k^{2k^3+k} n)$.
\endproofatend 

\hereorinappendixstatement
{
We remark that instead of using Corollary~\ref{cor:restricted:wl:solves:treewidth:no:smooth} we could also use the algorithm from~\cite{DBLP:journals/algorithmica/YamazakiBFT99} which decides for given root sets whether there is an isomorphism mapping the root sets to each other.
}

We remark that our algorithm for the theorem does not necessarily compute a connected root set. In fact the number of connected root sets that yield distance decompositions of smallest width cannot be bounded by an fpt function, and they in particular cannot be enumerated in fpt time. 

\ifhasappendix{ \marginpar{\scriptsize{\vspace{-1cm}{\color{gray}{The proof of this claim is in the Appendix.}}}}}\fi
\differingappendixstatement{This can be seen as follows.}{%
The following paragraphs show that the number of connected root sets that yield distance decompositions of optimal width cannot be bounded by a fixed-parameter tractable function.%
}%
\hereorinappendixstatement{We define the \emph{$(k,p)$-path} as the graph obtained from $P_{k}$
by replacing each edge with $p$ internally vertex-disjoint paths of length $2$ (see \figurename~\ref{fig:kp-path}).
We call the original vertices of $P_{k}$ \emph{black} and the others \emph{white}.
\ifhasappendix\else{The $(k+1,k)$-path is called $k$-path by Ding and Oporowski~\cite{DBLP:journals/dm/DingO96}. Among others, they use these graphs to characterize graphs of bounded strong tree width.}\fi

\begin{theorem}
 For every $k \ge 3$, there exists an infinite family of graphs such that each graph $G$
 in the family satisfies $\ctdw(G) = 2k-1$ and has $\left((n-k)/(k-1)\right)^{k-1}$ connected sets~$S$ for which~$\tdw_{S}(G) = 2k-1$, where~$n$ is the number of vertices of the graph.
\end{theorem}
\begin{proof}
 Let $G$ be the $(k, p)$-path where $p \ge 4k-2$.
 Let $S$ be the vertices on a shortest path between the leftmost and rightmost black vertices of~$G$.
 It is easy to see that there are $p^{k-1}$ such sets $S$ and $\tdw_{S}(G) = 2k-1$.
 Since $p = (n-k)/(k-1)$, it suffices now to show that $\ctdw(G) \ge 2k-1$.

 Assume that $\tdw_{S'}(G) < 2k-1 \le 4k-2$ for some connected set $S' \subseteq V(G)$.
 Then, by Lemma~\ref{lem:must:be:in:same:bag} and the assumption $p \ge 4k-2$, in any strong tree decomposition of width less than~$2k-1$, all black vertices have to be in the same bag.
 If they are in $S'$, then $G[S']$ contains a path between
 the leftmost and rightmost black vertices.
 This implies $|S'| \ge 2k-1$, a contradiction.
 Thus no black vertex can belong to $S'$.
 Now observe that every connected subgraph of $G$ without black vertices has only one vertex.
 Therefore it is impossible to put all the black vertices into the same bag.
\ifhasappendix \begin{figure}[h] \else \begin{figure}[t]\fi
\centering
\begin{tikzpicture}[scale=\scalefactor,thick]
  \def\xgap{3};
  \def\ygap{0.5};

  \foreach \x in {0,...,5}
  {
    \node (b[\x]) [circle,fill,draw] at (\xgap * \x, 0) {};
  }

  \foreach \x in {0,...,4}
  {
    \foreach \y in {0,...,3}
    {
      \node (w[\x][\y]) [circle,draw] at (\xgap * \x + \xgap/2, \ygap * \y - \ygap * 1.5) {};
    }
  }

  \foreach \x in {0,...,4}
  {
    \foreach \y in {0,...,3}
    {
     \draw (b[\x]) -- (w[\x][\y]);
     \pgfmathtruncatemacro{\tmp}{\x + 1};
     \draw (b[\tmp]) -- (w[\x][\y]);
    }
  }
\end{tikzpicture}

 \caption{The $(6,4)$-path.}
 \label{fig:kp-path}
\end{figure}
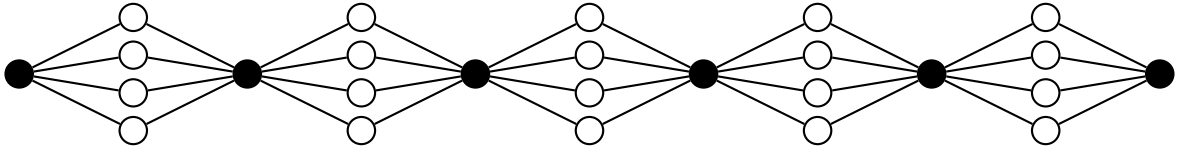%
\end{proof}

The fact that the number of initial sets that achieve a decomposition of optimal width is not bounded by a fixed-parameter tractable functions prevents our isomorphism algorithm from being directly used to compute all optimal initial sets in fixed-parameter tractable time. In fact our isomorphism algorithm does not determine the root-connected tree distance width of the input graphs.
}

\ifhasappendix
{It is also possible to show that the root-connected tree distance width of a graph cannot be bounded in terms of the rooted tree distance width (i.e., only one vertex in the root). \ifhasappendix{ \marginpar{\scriptsize{\vspace{-1cm}{\color{gray}{The proof is in the Appendix.}}}}}\fi
\differingappendixstatement{}
\fi
\hereorinappendixstatement{

The same family of graphs mentioned above can be used to observe that the root-connected tree distance width of a graph cannot be bounded in terms of the rooted tree distance width (i.e., when only one vertex in the root is allowed).}
\hereorinappendixstatement{
 \begin{theorem}
  For every $k \ge 2$, there is a graph~$G$ for which~$\ctdw(G) \le 5$ and $\rtdw(G) \ge k$.
 \end{theorem}
 \begin{proof}
  For any $p \ge 2k$, the $(3,p)$-path  is an example of such a graph.
 \end{proof}
}

Note that, in contrast to this, when we consider only path distance decompositions, the root-connected path distance width of a graph is bounded by a function of the rooted path distance width as shown in~\cite{Otachi2012ISAAC}.

\section{A reduction from strong tree width to maximum degree.}\label{sec:red:from:stw:to:deg}
\differingappendixstatement{}{\section{Proofs of Section~\ref{sec:red:from:stw:to:deg}}}

To define a reduction from isomorphism parameterized by strong tree width to isomorphism parameterized by maximum degree, we first reduce the problem to biconnected graphs relative to an equivalence relation and a suitable type of coloring compatible with the equivalence relation.

Let~$R$ be an equivalence relation on the vertices of a graph~$G$.
We define the \emph{quotient graph} of~$G$ with respect to~$R$ as the graph whose vertex set consists of the equivalence classes and in which two equivalence classes~$E_1$ and~$E_2$ are adjacent if there exists vertices~$v_1\in E_1$ and~$v_2\in E_2$ such that~$v_1$ and~$v_2$ are adjacent. We define \emph{the biconnected components of~$G$ relative to~$R$} as the sets of vertices that comprise biconnected components of the quotient graph, i.e., pre-images of biconnected components under the projection to the quotient graph. A \emph{tuple-coloring} is a map that assigns a color to every linear ordering of vertices in an equivalence class. When concerned with isomorphism of tuple-colored graphs equipped with an equivalence relation, we demand that isomorphisms preserve equivalence classes and the tuple-coloring, that is, an isomorphism must map an equivalence class to an equivalence class and colored ordered tuples to ordered tuples of the same color.

We define the \emph{biconnected component tree} (also called block-cut tree) of a graph~$G$ with respect to~$R$ as the following bipartite graph: the vertices of the one partition class are those equivalence classes that form cut-vertices in the quotient graph. The vertices of the other partition class are the biconnected components of~$G$ relative to~$R$. In the tree there is an edge between an equivalence class and a biconnected component if the corresponding cut-vertex is contained in the corresponding biconnected component in the quotient graph. 

The quotient graph can be constructed in time linear in the number of edges of a graph. Since the biconnected components of the quotient graph can then also be computed in polynomial time, the biconnected components and the component tree with respect to~$R$ can be computed in polynomial time.

\begin{lemma}\label{lem:reduction:to:biconnected}
The isomorphism problem of graphs with an equivalence relation on the vertices 
Turing-reduces to isomorphism of tuple-colored biconnected components of the input graphs relative to the equivalence relation. The running time is bounded by a polynomial in~$n$ and~$k!$, where~$n$ is the size of the input graphs and~$k$ is the size of the largest equivalence class.
\end{lemma}

\proofatend
We describe an algorithm that has access to an oracle that can perform isomorphism queries for tuple-colored biconnected components of the input graphs relative to the equivalence relation. 
Let~$G_1$ and~$G_2$ be two input graphs each equipped with an equivalence relation. Since we can perform isomorphism tests for pre-images of components of the quotient graph separately, without loss of generality we can assume that the quotient graphs are connected. Moreover, we can assume that each graph has a distinguished biconnected component and restrict ourselves to finding isomorphisms that match the two components. Indeed, by fixing distinguishing an arbitrary biconnected component of~$G_1$ and iterating over different possible choices for a distinguished component of~$G_2$, we will find an isomorphism, if there exists one. 
We use the distinguished biconnected components as roots for the biconnected component trees.

Our technique from here on essentially performs dynamic programming on these rooted trees.
More precisely, the algorithm proceeds as follows (details are given in Algorithm~\ref{algo:bicon:comp:red}). 
If both input graphs are already biconnected, we use the oracle to determine isomorphism. If one input graph is biconnected while the other one is not, then the graphs are non-isomorphic. Thus, we suppose both input graphs are not biconnected. Consider the leaves of the biconnected component trees. Note that these cannot be cut-vertex classes and thus they are biconnected components. For each leaf~$L$ we proceed as follows. Let~$V_L$ be the unique cut-vertex class of~$L$.
We assign to every ordering~$\sigma$ of~$V_L$ a color that depends only on the isomorphism type of the graph induced by~$L$ in which~$\sigma$ has a special color.
In other words, if two orderings of two cut vertex classes~$V_L$ and~$V_{L'}$ obtain the same color, then there is an isomorphism of the corresponding biconnected components~$L$ and~$L'$ which maps the cut vertex classes to each other while observing their ordering. More formally we compute a function~$\chi$ such that for leaves~$L$ and~$L'$ with orders~$\sigma$ and~$\sigma'$, respectively, we have~$\chi(L,\sigma) = \chi(L',\sigma')$ if and only if said isomorphism exists.
To determine this isomorphism type we use the oracle. Since we may have to check all pairs of leaves, this may require a number of oracle calls quadratic in the number of colored graphs whose isomorphism types are to be determined. 

For every cut-vertex class adjacent to a leaf, we assign to every ordering of this cut-vertex class a new color that depends on its previous color (in case the ordering was already colored) and the multiset of colors computed for this ordering for all leaves. More precisely, the color depends on the multi-set~$\{\!\{\chi(L,\sigma) \mid L \text{ a leaf and } \sigma \text{ is an ordering of } V_L \}\!\}$.

We then remove all vertices contained in leaves but not contained in cut-vertices from the input graphs.

To prevent the names of colors from becoming excessively long strings, we always rename newly arising colors by assigning previously unused integers to them.

After the modification, the graphs are isomorphic as tuple-colored graphs if and only if they were isomorphic before.
By repeating the process at most a linear number of times the process will terminate with one of the graphs being biconnected in the last iteration.

Since in each iteration the modified graphs are isomorphic if and only if the original graphs are isomorphic, the algorithm correctly determines isomorphism. It remains to bound the running time. Since all steps between oracle calls can be performed in time polynomial in~$n$ and~$k!$, it suffices to bound the number of oracle calls. The number of these is at most the number of pairs of biconnected components with distinguished cut-vertex classes equipped with a linear ordering. The number of biconnected components and cut-vertex classes is polynomial in~$n$ while the number of orderings of a cut-vertex class is at most~$k!$. 
\endproofatend
\hereorinappendixstatement{
\begin{algorithm}[t]
\caption{A reduction to biconnected components relative to an equivalence relation $\isored(G_1,G_2,R_1,R_2,B_1,B_2)$}
\label{algo:bicon:comp:red}
\begin{algorithmic}[1]
\REQUIRE Two tuple-colored graphs~$G_1$ and~$G_2$ with equivalence relations~$R_1$ and~$R_2$, such that the quotient graphs are connected and for each a distinguished biconnected component~$B_1$ and~$B_2$, respectively.
An Oracle~$\Orc$ that determines isomorphism of biconnected graphs.
\ENSURE \textbf{Yes} if there is an isomorphism from~$G_1$ to~$G_2$ that maps~$B_1$ to~$B_2$, \textbf{No} otherwise.
\ENSUREGAP

\WHILE {both graphs~$G_1$ and~$G_2$ are not biconnected relative to~$R_1$ and $R_2$, respectively} 

	\STATE compute the biconnected component trees~$T_1$ and~$T_2$ with roots~$B_1$ and~$B_2$
	\FORALL {leaves~$L$}
	 \STATE compute the unique cut-vertex class~$V_L$ adjacent to~$L$
	 \ENDFOR
	 \STATE compute~$M = \{(L,\sigma)\mid L \text{ is a leaf and } \sigma \text{ a linear order of } V_L\}$ 
	 \FORALL {pairs of~$\{((L_1,\sigma_1),(L_2,\sigma_2))\mid (L_i,\sigma_i)\in M\}$}
	 \STATE use~$\Orc$ to decide if there exists an isomorphism from~$L_1$ to~$L_2$ mapping~$V_{L_1}$ to~$V_{L_2}$ respecting the orders~$\sigma_i$
	  \ENDFOR 
	  \STATE compute an isomorphism invariant coloring~$\chi(L,\sigma)$  for every~$(L,\sigma)\in M$
	  \STATE 
	    assign to every ordering~$\sigma$ of a cut-vertex class~$V$ adjacent to a leaf a previously unused color that depends on the multi-set~$\{\!\{\chi(L,\sigma)\mid  (L,\sigma) \in M\}\!\}$
	 \STATE remove all leafs~$L$ from both graphs
\ENDWHILE

\IF {both graphs~$G_1$ and~$G_2$ are biconnected relative to~$R_1$ and $R_2$ respectively}
	\STATE\textbf{return}~$\Orc(G_1,G_2,B_1,B_2)$.
\ELSE
	\STATE \textbf{return}~\textbf{false}
\ENDIF
\end{algorithmic}
\end{algorithm}
}

From the theorem we obtain as corollary that to decide isomorphism of graphs in a hereditary graph class, i.e., a class closed under taking induced subgraphs, it suffices to be able to decide isomorphism of biconnected vertex-colored graphs.

\begin{corollary}\label{cor:red:to:biconnected}
The graph isomorphism problem of vertex-colored graphs in a hereditary graph class~$\mathcal{C}$ polynomial-time Turing-reduces to the isomorphism problem of biconnected vertex-colored graphs in~$\mathcal{C}$.
\end{corollary}
\proofatend
The class of colored graphs is the same as the class of~\emph{tuple-colored} graphs where the equivalence relation is equality.
Thus, the corollary follows from Lemma~\ref{lem:reduction:to:biconnected} since the algorithm described in its proof does not alter the equivalence relation of the input graphs.
\endproofatend

For general graph isomorphism, for every integer~$k$, it is possible to reduce the isomorphism problem to isomorphism of~$k$-connected graphs by simply adding universal vertices adjacent to all other vertices. However, for hereditary graph classes, under application of this technique or similar gadget constructions, the graphs may not necessarily remain within the class. In fact, if there is reduction  to 3-connected graphs analogous to Corollary~\ref{cor:red:to:biconnected}, then graph isomorphism would be polynomial-time solvable in general. This can be seen by considering the isomorphism-complete class of bipartite graphs in which in one bipartition class every vertex has degree at most 2. This class does not contain any 3-connected graphs that have components with more than 2 vertices. 

We will now employ the relation~$\kcon_{2k}$ defined in Section~\ref{sec:distance:decomp}. However, this relation is not necessarily an equivalence relation.
Let~$\kcon_{2k}^{+}$ be the transitive closure of the relation~$\kcon_{2k}$.
It turns out that graphs that are biconnected relative to~$\kcon_{2k}^{+}$ have bounded degree.
\begin{lemma}\label{lem:bd:stw:and:bicon:imply:bd:degree}
For~$k\geq 2$, if a graph~$G$ with~$\stw(G)\leq k$ is biconnected relative to~$\kcon_{2k}^{+}$ then~$G$ has a maximum degree of at most~$ 2k^2 (k-1)+k-1$.
\end{lemma}
\proofatend
The method of this proof is similar to the proof of Lemma~\ref{lem:no:highly:connected:vertex:pair:implies:bd:degree}. Suppose~$v$ is a vertex in~$G$ of degree~$\deg(v)$ larger than~$ 2k^2 (k-1)+k-1$. Let~$\mathcal{S}$ be a strong tree-decomposition of~$G$ of width at most~$k$.
By Lemma~\ref{lem:must:be:in:same:bag}, vertices equivalent under~$\kcon_{2k}$ must be in the same bag. Consequently vertices that are equivalent under~$\kcon_{2k}^{+}$ must also be in the same bag.
 Since the degree of~$v$ is at least~$2k$, the bag of~$v$ is adjacent to at least two bags which each contain a vertex~$u_1$ and~$u_2$ say.
Moreover, the bag that contains~$v$ separates the vertices~$u_1$ and~$u_2$. Since the quotient graph is biconnected relative to~$\kcon_{2k}^{+}$, from every vertex adjacent to but not equivalent to~$v$ under~$\kcon_{2k}^{+}$ there is a path to~$u_1$ and a path to~$u_2$ that does not include any vertex equivalent to~$v$ under~$\kcon_{2k}^{+}$. One of the paths must include a vertex from the bag of~$v$.

Thus, for every vertex~$w$ adjacent to but not in the same bag as~$v$, there is a path to a vertex contained in the bag of~$v$ that is not equivalent to~$v$. By taking shortest paths that fulfill these conditions, we achieve that for neighbors of~$v$ in different bags, these paths are internally vertex-disjoint.
At most~$k-1$ vertices are contained in the same bag as~$v$. Thus, since a bag contains at most~$k$ vertices, there are at least~$(\deg(v)-(k-1))/k$ internally vertex disjoint paths from~$v$ to vertices which are in the same bag as~$v$ but not equivalent to~$v$.
Since~$(|\deg(v)|-(k-1))/k >  (k-1) 2k$, by the pigeonhole principle, this shows the existence of a vertex in the same bag as~$v$ not equivalent to~$v$~$\kcon_{2k}^{+}$ from which there are~$2k$ internally vertex-disjoint paths to~$v$, yielding a contradiction. 
\endproofatend

\begin{lemma}\label{lem:replace:coloring:with:gadgets}
The isomorphism problem of tuple-colored graphs of degree at most~$d$ with an equivalence relation on the vertices with no equivalence class having more than~$k$ elements reduces to isomorphism of uncolored graphs of degree at most~$\BigO(d+k!)$. The running time is polynomial in~$n$ and~$k!$, where~$n$ is the size of the input graphs.

\end{lemma}
\proofatend
We describe a reduction that satisfies the running time requirements to vertex-colored graphs. A standard reduction of attaching trees of equal height encoding vertex colors then reduces the problem further to uncolored graphs (see~\cite{SchweitzerThesis}).

Let~$C$ be an equivalence class. For every ordering of~$C$, we attach a gadget to the vertices of~$C$ encoding the color of this ordering (see Figure~\ref{fig:der:kamm}). We add a path of newly added  vertices~$v_0,v_1,\ldots,v_{|C|}$. Vertex~$v_0$ is colored with the color of the ordering of~$C$. For~$i\in \{1,\ldots,n\}$ an edge from~$v_i$ to the~$i$-th vertex of~$C$ under the ordering is added. 
If this procedure is applied to two graphs, the resulting graphs are isomorphic if and only if the original graphs are.
The maximum degree within the gadgets is~$3$, while the number of new edges added to a vertex originally in the graph is at most~$k!$, where~$k$ is the size of the largest equivalence class.
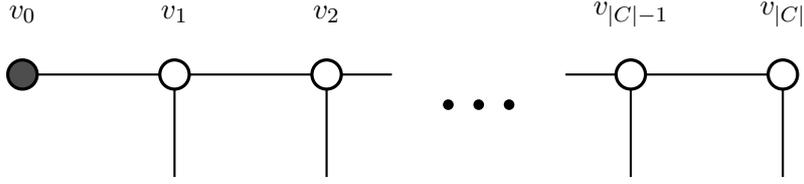
\begin{figure}[t]
\centering
\begin{tikzpicture}[scale=\scalefactor,thick]
\node (v0) at (-2,0) {};
\node (v1) at (0,0) {};
\node (v2) at (2,0) {};
\node (v3l) at (3,0) {};
\node (v3r) at (5,0) {};
\node (v3a) at (3.6,-0.4) {};
\node (v3b) at (  4,-0.4) {};
\node (v3c) at (4.4,-0.4) {};
\node (v4) at (6,0) {};
\node (v5) at (8,0) {};

\node (u1) at (0,-1.5) {};
\node (u2) at (2,-1.5) {};
\node (u3) at (4,-1.5) {};
\node (u4) at (6,-1.5) {};
\node (u5) at (8,-1.5) {};

\draw (v0) -- (v1) -- (v2) -- (v3l);
\draw (v3r)-- (v4)-- (v5);
\draw (v1) -- (u1);
\draw (v2) -- (u2);
\draw (v4) -- (u4);
\draw (v5) -- (u5);

\node[grayvertex] at (v0)  {};
\node[vertex] at (v1)  {};
\node[vertex] at (v2)  {};
\node[circle,fill=black, minimum size=4pt,inner sep= 0pt] at (v3a)  {};
\node[circle,fill=black, minimum size=4pt,inner sep= 0pt] at (v3b)  {};
\node[circle,fill=black, minimum size=4pt,inner sep= 0pt] at (v3c)  {};
\node[vertex] at (v4)  {};
\node[vertex] at (v5)  {};

\draw (-2,0.8) node {$v_0$};
\draw (0,0.8) node {$v_1$};
\draw (2,0.8) node {$v_2$};
\draw (6,0.8) node {$v_{|C|-1}$};
\draw (8,0.8) node {$v_{|C|}$};
\end{tikzpicture}

\caption{The gadget encoding the color of a linear ordering of a set~$C$. The vertex~$v_0$ is colored in the color that was associated with the linear ordering of~$C$.}\label{fig:der:kamm}
\end{figure}%
\endproofatend

For the lemma, it is essential that equivalence classes do not intersect. By coloring partially overlapping sets, it is possible to encode hypergraphs, and thus to encode graphs, even with sets of size at most~2.\ifhasappendix\else{Consequently a reduction for this case would yield a polynomial time algorithm for graph isomorphism.}\fi%
\hereorinappendixstatement{%
The construction in the previous proof yields a constructive method for representing a given permutation group, which must be possible by Frucht's Theorem~\appcite{frucht}. In the meantime, extensive research on representablility of groups as automorphism groups of graphs has been conducted (see \appcite{MR797250}).%
}%
\ifhasappendix\else

\fi%

Together, the lemmas of this section can be used to assemble a reduction from strong tree width to maximum degree.

\begin{theorem}\label{thm:strong:tw:to:max:deg}
There is an fpt Turing-reduction from isomorphism parameterized by strong tree width to isomorphism parameterized by maximum degree.
\end{theorem}
\begin{proof}
We first compute~$\kcon_{2k}$ and the transitive closure~$\kcon_{2k}^{+}$ which can be done in polynomial time. 
If the largest equivalence class has size greater than~$k$ we reject the input as infeasible containing a graph of strong tree width larger than~$k$.
We then reduce via Lemma~\ref{lem:reduction:to:biconnected} to biconnected graphs relative to~$\kcon_{2k}^{+}$. By Lemma~\ref{lem:bd:stw:and:bicon:imply:bd:degree} the biconnected components have bounded degree. By Lemma~\ref{lem:replace:coloring:with:gadgets} we can then reduce the isomorphism problem of the tuple-colored biconnected components to isomorphism of graphs of bounded degree.
\end{proof}

\section{Applications to fpt isomorphism results}

In this section we combine the two results from the previous sections to obtain further fpt isomorphism algorithms.

\begin{theorem}\label{thm:con:stg:tw:fpt}
Graph isomorphism parameterized by connected strong tree width can be solved in fpt time.
\end{theorem}
\begin{proof}
By Theorem~\ref{thm:strong:tw:to:max:deg} the problem reduces to isomorphism of the tuple-colored biconnected components. 
Let~$G$ be a graph of connected strong tree width at most~$k$ of bounded degree. To apply Corollary~\ref{cor:restricted:wl:solves:treewidth:no:smooth} we first describe a set of possible bags capturing a strong tree decomposition of~$G$ computable in fpt time and having an fpt size bound. 
For this consider the family~$\mathcal{V}(G)$ of sets of size at most~$k$ that project to a connected subgraph 
in the quotient graph relative to~$\kcon_{2k}$. Since 
the degree of~$G$ is bounded, the number of such sets is bounded by an fpt number.   
The theorem now follows from Corollary~\ref{cor:restricted:wl:solves:treewidth:no:smooth}.
\end{proof}

We have shown fixed-parameter tractability for isomorphism with respect to the parameters connected strong tree width and root-connected tree distance width. 
In turns out that these parameters are unrelated, i.e., that neither of these parameters can be bounded by a function of the other. 
\ifhasappendix{ \marginpar{\scriptsize{\vspace{-1cm}{\color{gray}{The proof of this claim is in the Appendix.}}}}}\fi

\hereorinappendixstatement{
\ifhasappendix{The following two theorems show that the parameters connected strong tree width and root-connected tree distance width are unrelated.}\fi
\begin{theorem}
 For every $k \ge 2$, there is a graph~$G$ for which $\ctdw(G) \le 2$ and $\cstw(G) \ge k$.
\end{theorem}
\begin{proof}
 The cycle $C_{2k}$ on~$2k$ vertices has these properties.
\end{proof}

 The \emph{$(k,p)$-comb} is the graph obtained from the path $P_{k}$
 by attaching a copy of~$K_{2,p}$ to each vertex of $P_{k}$ as depicted in \figurename~\ref{fig:kp-comb}.
 We call the vertices of degree~$2$ \emph{white} and the others \emph{black}.
 It is easy to see that for every $k$ and $p$, the $(k,p)$-comb has connected strong tree width at most $3$
 (see \figurename~\ref{fig:kp-comb-cstd}).
\begin{theorem}
 For every $k \ge 3$, there is a graph~$G$ for which~$\cstw(G) \le 3$ and~$\ctdw(G) \ge k$.
\end{theorem}
\begin{proof}
 Let $G$ be the $(k, p)$-comb where $p \ge 2k$.
 Suppose that $\tdw_{S}(G) < k$ for some connected set $S \subseteq V(G)$.
 Since $|S| < k$, there is a copy of $K_{2,p}$ in~$G$ that does not intersect $S$.
 Therefore, the two black vertices in that copy cannot be in the same bag.
 This contradicts Lemma~\ref{lem:must:be:in:same:bag}.
\end{proof}
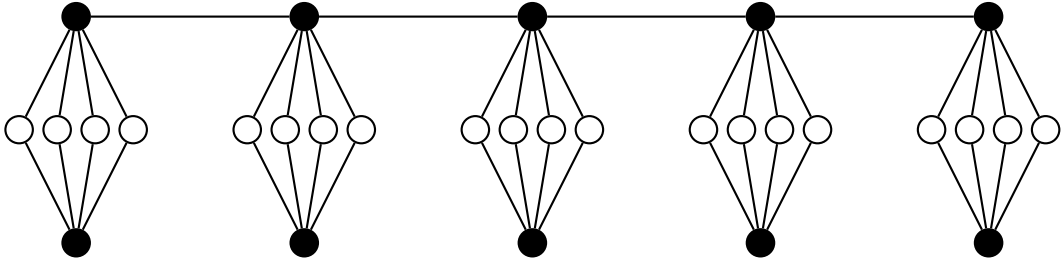
\begin{figure}[t]
\centering
\begin{tikzpicture}[scale=\scalefactor,thick]
  \def\xgap{3};
  \def\ygap{0.5};

  \foreach \x in {0,...,4}
  {
    \node (b[0][\x]) [circle,fill,draw] at (\xgap * \x, 0) {};
    \node (b[1][\x]) [circle,fill,draw] at (\xgap * \x, \xgap) {};
  }

  \foreach \x in {0,...,4}
  {
    \foreach \y in {0,...,3}
    {
      \node (w[\x][\y]) [circle,draw] at (\xgap * \x + \ygap * \y - 1.5 * \ygap, \xgap/2) {};
    }
  }

  \foreach \x in {0,...,4}
  {
    \foreach \y in {0,...,3}
    {
      \draw (b[0][\x]) -- (w[\x][\y]);
      \draw (b[1][\x]) -- (w[\x][\y]);
    }
  }
  \foreach \x in {0,...,3}
  {
    \pgfmathtruncatemacro{\tmp}{\x + 1};
    \draw (b[1][\x]) -- (b[1][\tmp]);
  }
\end{tikzpicture}

 \caption{The $(5,4)$-comb.}
 \label{fig:kp-comb}
\end{figure}%

\begin{figure}[t]
\centering
\begin{tikzpicture}[scale=\scalefactor,thick]
  \def\xgap{3};
  \def\ygap{.75};

  \foreach \x in {0,...,4}
  {
    \node (b[0][\x]) [circle,fill,draw] at (\xgap * \x - .4, \xgap) {};
    \node (b[1][\x]) [circle,fill,draw] at (\xgap * \x + .4, \xgap + .6) {};
  }

  \foreach \x in {0,...,4}
  {
    \node (w[\x][0]) [circle,draw] at (\xgap * \x, \xgap + .3) {};
    \node[shape=rectangle, rounded corners, draw, minimum height=30pt, minimum width=40pt]
      (ibag[\x]) at (w[\x][0]) {};
  }
  \foreach \x in {0,...,4}
  {
    \foreach \y in {1,...,3}
    {
      \node (w[\x][\y]) [circle,draw] at (\xgap * \x + \ygap * \y - 2 * \ygap, \xgap/2) {};
      \node[shape=rectangle, rounded corners, draw, minimum size=16pt]
        (lbag[\x][\y]) at (w[\x][\y]) {};
    }
  }

 \foreach \x in {0,...,3}
  {
    \pgfmathtruncatemacro{\tmp}{\x + 1};
    \draw[line width=2pt] (ibag[\x]) -- (ibag[\tmp]);
  }
  \foreach \x in {0,...,4}
  {
    \foreach \y in {1,...,3}
    {
      \draw[line width=2pt] (ibag[\x]) -- (lbag[\x][\y]);
    }
  }

  \foreach \x in {0,...,4}
  {
    \foreach \y in {0,...,3}
    {
      \draw[color=blue,dashed] (b[0][\x]) -- (w[\x][\y]);
      \draw[color=blue,dashed] (b[1][\x]) -- (w[\x][\y]);
    }
  }
  \foreach \x in {0,...,3}
  {
    \pgfmathtruncatemacro{\tmp}{\x + 1};
    \draw[color=blue,dashed] (b[1][\x]) -- (b[1][\tmp]);
  }
\end{tikzpicture}

 \caption{A connected strong tree decomposition of the $(5,4)$-comb of width 3.}
 \label{fig:kp-comb-cstd}
\end{figure}
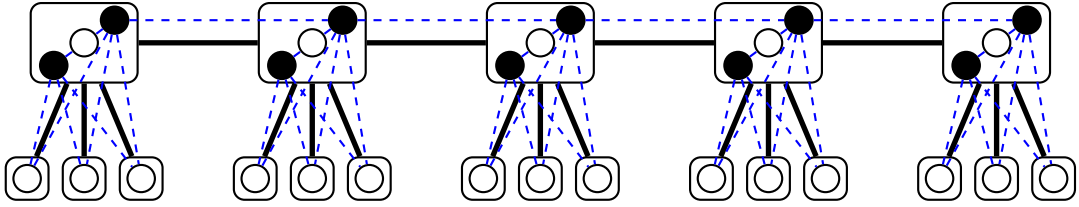%
}

\hereorinappendixstatement{\ifhasappendix{The following lemma is required for the proof of Theorem~\ref{thm:fpt:geod:cyc:tw}.}\fi
\begin{lemma}\label{lem:bound:on:geod:impl:bd:dist}
Let~$G$ be a connected graph with strong tree width at most~$k$ and for which the longest geodesic cycle has length at most~$c$. There exists a strong tree decomposition of width at most~$k$ with the following property: If~$S$ is a proper subset of a bag~$B$ then there is a vertex~$v$ in~$B\setminus S$ that is of distance at most~$c$ from~$S$.
\end{lemma}

\begin{proof}
Let~$G$ be a graph with strong tree width at most~$k$. Consider a  strong tree decomposition of width at most~$k$ for which the bag partition is finest. This implies that the components of the forest obtained by removing a bag induce connected graphs. We show that such a decomposition has the desired property.
Thus let~$S$ be a proper subset of a bag~$B$. If~$S$ has a neighbor in~$B$ this claim is obvious, so we suppose otherwise. If~$B$ were a leaf, then~$B$ would be connected, so we also suppose~$B$ is not a leaf. In this case, let~$B'$ be a neighbor of~$B$ whose subtree connects a vertex~$s$ from~$S$ with some vertex~$v$ in~$B$ outside of~$S$. In the subtree of~$B$ in the graph in which~$B'$ has been deleted, there is a path from~$s$ to~$v$. This implies there is a cycle containing a vertex in~$S$ and a vertex of~$B$ not in~$S$. Let~$C$ be a shortest such cycle. We claim that this cycle is geodesic. If~$C$ were not geodesic, there would exist a shortcut connecting two vertices~$c_1$ and~$c_2$ of~$C$ via a path~$P$ shorter than the distance from~$c_1$ to~$c_2$ on the cycle. Since~$C$ was chosen of minimal length, this implies that the path~$P$ contains a vertex~$x$ of~$B$. Together with two segments of~$C$ the path~$P$ forms two cycles shorter than~$C$. Since~$P$ contains a vertex from~$B$, one of the cycles contains both a vertex from~$S$ and~$B\setminus S$.  This contradicts minimality of~$C$. 
\end{proof}
}
As further examples of applications of our technique we can obtain fixed-parameter tractability results of other parameters as follows.
A geodesic cycle in a graph~$G$ is a cycle~$C$ such that the distance between every two vertices in~$C$ is the same as the distance in~$G$. The chordality of a graph is the length of the longest induced cycles.

\begin{theorem}\label{thm:fpt:geod:cyc:tw}
Graph isomorphism parameterized by the maximum of the length of a geodesic cycle and strong tree width can be solved in fpt time.
\end{theorem}
\proofatend
For a biconnected component~$C$ of~$G$ relative to~$\kcon_{2k}$ we let $\mathcal{V}_C$ be the family of vertex sets~$B$ that have the property that for every subset~$S\subset B$ there are vertices~$s\in S$ and~$v\in B\setminus S$ such that~$s$ and~$v$ are equivalent or~$s$ and~$v$ are of distance at most~$c$ where~$c$ is the length of the longest geodesic cycle of~$G$. By the previous lemma this family contains all bags of a strong tree decomposition of~$C$ and by Lemma~\ref{lem:bd:stw:and:bicon:imply:bd:degree} this set is of fpt size. Moreover it can be computed in fpt time.
We conclude the theorem with the same arguments already used in the proof of Theorem~\ref{thm:con:stg:tw:fpt}.
\endproofatend

\begin{corollary}
Graph isomorphism parameterized by the maximum of the chordality and degree can be solved in fpt time.
\end{corollary}
\proofatend
The corollary follows from the previous theorem by noting that bounded chordality implies a bound on the length of geodesic cycles.
Moreover if the chordality and the degree is bounded, then the graph has bounded tree width~\appcite{DBLP:journals/dam/BodlaenderT97}, \appcite{DBLP:conf/icalp/KosowskiLNS12} 
and consequently it has bounded strong  tree width  \cite{MR1358539},~\appcite{DBLP:journals/ejc/Wood09}.
\endproofatend

Further applications of the theorems can be obtained by considering the set of potential maximal cliques. A \emph{potential maximal clique} of a graph~$G$ is a set of vertices that is a bag in some minimal tree-decomposition of~$G$.
The set of potential maximal cliques can be computed in polynomial time in the size of the set itself~\cite{DBLP:journals/tcs/BouchitteT02}. 
Moreover, this set is isomorphism invariant and the subset of potential maximal cliques of size~$k$ captures a tree decomposition of~$G$ of minimal width. We can thus apply Corollary~\ref{cor:restricted:wl:solves:treewidth:no:smooth} to the potential maximal cliques.

\begin{theorem} If for a parameterized graph class there is an fpt bound on the number of potential maximal bags then isomorphism is fixed parameter tractable in the maximum of the parameter and the tree width.
\end{theorem}

\proofatend
By Corollary~\ref{cor:restricted:wl:solves:treewidth:no:smooth}, it suffices to compute for all graphs an invariant fpt size family of sets capturing a tree decomposition. 

The set of potential maximal cliques captures a tree decomposition, is isomorphism invariant, can be computed in time polynomial in its size~\cite{DBLP:journals/tcs/BouchitteT02} and is bounded by an fpt function by assumption. It thus satisfies the requirements.
\endproofatend

Equivalently, in the theorem it suffices to have a bound on the number of minimal $s$-$t$-separators, since by a theorem of  Bouchitt{\'e} and Todinca~\cite{DBLP:journals/tcs/BouchitteT02} the number of potential maximal cliques is polynomially bounded the number of minimal $s$-$t$-separators.

Using results from~\cite{DBLP:conf/soda/FominTV14} we can apply the theorem to various graph classes. Indeed, it is known that the number of minimal $s$-$t$-separators is polynomially bounded for weakly chordal, polygonal circle, circular-arc and~$d$-trapezoid graphs (see~\cite[Section 5]{DBLP:conf/soda/FominTV14}). So for all these classes, isomorphism is fixed parameter tractable when parameterized by tree width. Moreover, for weakly chordal graphs and circular-arc graphs the tree width of~$H$-minor free graphs is bounded by a function of the number of vertices in~$H$~(see also \cite[Section 5]{DBLP:conf/soda/FominTV14}), so isomorphism of~$H$-minor free weakly chordal graphs  and~$H$-minor free circular-arc graphs is fixed parameter tractable when parameterized by the size of~$H$.
\section{Conclusion}

In this paper 
we show that, in order to perform isomorphism tests, it suffices to compute an invariant set of potential bags that is comprehensive enough to express a tree decomposition. Indeed, by applying the restricted Weisfeiler-Lehman, we do not need to worry about how to perform the isomorphism test, nor how to compute a decomposition. 

For various other results this means that their isomorphism testing part can be replaced by the general theorem and only the part analyzing the graph class remains. For example in~\cite{DBLP:conf/isaac/Nagoya01} and~\cite{doi:10.1137/0603025} it is shown that for chordal graphs of bounded clique number a tree model can be computed in cubic time. This tree model is unique up to the ordering of children and gives rise to an invariant family of sets of vertices capturing a tree-decomposition.

In~\cite{DBLP:conf/iwpec/BoulandDK12} it is shown that for graphs of tree-depth at most~$k$ the number of vertices that can be chosen as the root in a tree-depth decomposition is bounded by a function of~$k$, recursively applying this gives rise to invariant family of sets of vertices capturing a tree-decomposition.

Furthermore we demonstrated how the theorems can be used in conjunction with the set of potential maximal cliques. The advantage here is that it is known in general that this set can be efficiently computed.

However, our technique can also be applied to tree decompositions, where the long question whether graph isomorphism is fixed parameter tractable when parameterized by tree width remains.

On the other hand, our parameterized reduction to bounded degree is only valid for strong tree width, and whether such a reduction exists for tree width remains open. Finally, as mentioned in the introduction, no non-tractability results are known in this context, and parameterized reduction could be a method to establish a hardness criteria.

\bibliographystyle{abbrv}
\bibliography{main}

\begin{thebibliography}{10}

\bibitem{DBLP:conf/fsttcs/ArvindDKT10}
V.~Arvind, B.~Das, J.~K{\"o}bler, and S.~Toda.
\newblock Colored hypergraph isomorphism is fixed parameter tractable.
\newblock In {\em FSTTCS}, pages 327--337, 2010.

\bibitem{DBLP:conf/esa/BerkholzBG13}
C.~Berkholz, P.~Bonsma, and M.~Grohe.
\newblock Tight lower and upper bounds for the complexity of canonical colour
  refinement.
\newblock In {\em ESA}, pages 145--156, 2013.

\bibitem{DBLP:journals/siamcomp/Bodlaender96}
H.~L. Bodlaender.
\newblock A linear-time algorithm for finding tree-decompositions of small
  treewidth.
\newblock {\em SIAM J. Comput.}, 25(6):1305--1317, 1996.

\bibitem{DBLP:journals/jco/Bodlaender03}
H.~L. Bodlaender.
\newblock Necessary edges in k-chordalisations of graphs.
\newblock {\em J. Comb. Optim.}, 7(3):283--290, 2003.

\bibitem{DBLP:journals/dam/BodlaenderT97}
H.~L. Bodlaender and D.~M. Thilikos.
\newblock Treewidth for graphs with small chordality.
\newblock {\em Discrete Applied Mathematics}, 79(1--3):45--61, 1997.

\bibitem{DBLP:journals/tcs/BouchitteT02}
V.~Bouchitt{\'e} and I.~Todinca.
\newblock Listing all potential maximal cliques of a graph.
\newblock {\em Theor. Comput. Sci.}, 276(1-2):17--32, 2002.

\bibitem{DBLP:conf/iwpec/BoulandDK12}
A.~Bouland, A.~Dawar, and E.~Kopczynski.
\newblock On tractable parameterizations of graph isomorphism.
\newblock In {\em IPEC}, pages 218--230, 2012.

\bibitem{cfi:paper}
J.-Y. Cai, M.~F\"{u}rer, and N.~Immerman.
\newblock An optimal lower bound on the number of variables for graph
  identification.
\newblock {\em Combinatorica}, 12(4):389--410, 1992.

\bibitem{MR797250}
P.~J. Cameron.
\newblock Automorphism groups of graphs.
\newblock In {\em Selected topics in graph theory, 2}, pages 89--127. Academic
  Press, London, 1983.

\bibitem{DBLP:conf/coco/DattaLNTW09}
S.~Datta, N.~Limaye, P.~Nimbhorkar, T.~Thierauf, and F.~Wagner.
\newblock Planar graph isomorphism is in log-space.
\newblock In {\em IEEE Conference on Computational Complexity}, pages 203--214,
  2009.

\bibitem{MR1358539}
G.~Ding and B.~Oporowski.
\newblock Some results on tree decomposition of graphs.
\newblock {\em J. Graph Theory}, 20(4):481--499, 1995.

\bibitem{DBLP:journals/dm/DingO96}
G.~Ding and B.~Oporowski.
\newblock On tree-partitions of graphs.
\newblock {\em Discrete Math.}, 149(1--3):45--58, 1996.

\bibitem{Evdokimov}
S.~Evdokimov and I.~N. Ponomarenko.
\newblock Isomorphism of coloured graphs with slowly increasing multiplicity of
  jordan blocks.
\newblock {\em Combinatorica}, 19(3):321--333, 1999.

\bibitem{Flum2006}
J.~Flum and M.~Grohe.
\newblock {\em Parameterized Complexity Theory (Texts in Theoretical Computer
  Science. An EATCS Series)}.
\newblock Springer, London, UK, March 2006.

\bibitem{DBLP:conf/soda/FominTV14}
F.~V. Fomin, I.~Todinca, and Y.~Villanger.
\newblock Large induced subgraphs via triangulations and cmso.
\newblock In {\em SODA}, pages 582--583, 2014.

\bibitem{frucht}
R.~Frucht.
\newblock {Herstellung von Graphen mit vorgegebener abstrakter Gruppe.}
\newblock {\em Compos. Math.}, 6:239--250, 1938.

\bibitem{fuhlbrueck}
F.~Fuhlbr\"{u}ck.
\newblock {Fixed-parameter tractability of the graph isomorphism and
  canonization problems}.
\newblock Diploma thesis, Humboldt-Universit\"{a}t zu Berlin, 2013.

\bibitem{DBLP:conf/focs/FurstHL80}
M.~L. Furst, J.~E. Hopcroft, and E.~M. Luks.
\newblock Polynomial-time algorithms for permutation groups.
\newblock In {\em FOCS}, pages 36--41, 1980.

\bibitem{DBLP:conf/lics/Grohe10}
M.~Grohe.
\newblock Fixed-point definability and polynomial time on graphs with excluded
  minors.
\newblock In {\em LICS}, pages 179--188, 2010.

\bibitem{DBLP:conf/icdt/GroheM99}
M.~Grohe and J.~Mari{\~n}o.
\newblock Definability and descriptive complexity on databases of bounded
  tree-width.
\newblock In {\em ICDT}, pages 70--82, 1999.

\bibitem{compleretro}
N.~Immerman and E.~Lander.
\newblock {\em Complexity Theory Retrospective}, chapter Describing Graphs: A
  First-Order Approach to Graph Canonization.
\newblock Springer-Verlag, 1990.

\bibitem{doi:10.1137/0603025}
M.~Klawe, D.~Corneil, and A.~Proskurowski.
\newblock Isomorphism testing in hookup classes.
\newblock {\em SIAM Journal on Algebraic Discrete Methods}, 3(2):260--274,
  1982.

\bibitem{KoeblerSchoeningToran:1993}
J.~K\"{o}bler, U.~Sch\"{o}ning, and J.~Tor\'{a}n.
\newblock {\em The graph isomorphism problem: its structural complexity}.
\newblock Birkh\"{a}user Verlag, Basel, Switzerland, Switzerland, 1993.

\bibitem{DBLP:conf/icalp/KosowskiLNS12}
A.~Kosowski, B.~Li, N.~Nisse, and K.~Suchan.
\newblock k-chordal graphs: From cops and robber to compact routing via
  treewidth.
\newblock In {\em ICALP (2)}, pages 610--622, 2012.

\bibitem{DBLP:conf/swat/KratschS10}
S.~Kratsch and P.~Schweitzer.
\newblock Isomorphism for graphs of bounded feedback vertex set number.
\newblock In {\em SWAT}, pages 81--92, 2010.

\bibitem{DBLP:conf/isaac/Nagoya01}
T.~Nagoya.
\newblock Counting graph isomorphisms among chordal graphs with restricted
  clique number.
\newblock In {\em ISAAC}, pages 136--147, 2001.

\bibitem{Otachi2012ISAAC}
Y.~Otachi.
\newblock Isomorphism for graphs of bounded connected-path-distance-width.
\newblock In {\em ISAAC}, pages 455--464, 2012.

\bibitem{SchweitzerThesis}
P.~Schweitzer.
\newblock {\em Problems of unknown complexity: graph isomorphism and Ramsey
  theoretic numbers}.
\newblock Phd thesis, Universit\"{a}t des {S}aarlandes, Saarbr\"{u}cken,
  Germany, July 2009.

\bibitem{DBLP:conf/esa/Schweitzer11}
P.~Schweitzer.
\newblock Isomorphism of (mis)labeled graphs.
\newblock In {\em ESA}, pages 370--381, 2011.

\bibitem{Seese:1985:TGC:647892.739763}
D.~Seese.
\newblock Tree-partite graphs and the complexity of algorithms.
\newblock In {\em Fundamentals of Computation Theory}, FCT '85, pages 412--421,
  1985.

\bibitem{Toda:2001:GIDP}
S.~Toda.
\newblock {\em Gurafu Doukeisei Hantei Mondai (The Graph Isomorphism Decision
  Problem)}.
\newblock Nihon University, Tokyo, Japan, 2001.
\newblock in Japanese.

\bibitem{DBLP:journals/ieicet/Toda06}
S.~Toda.
\newblock Computing automorphism groups of chordal graphs whose simplicial
  components are of small size.
\newblock {\em IEICE Transactions}, 89-D(8):2388--2401, 2006.

\bibitem{DBLP:journals/ejc/Wood09}
D.~R. Wood.
\newblock On tree-partition-width.
\newblock {\em Eur. J. Comb.}, 30(5):1245--1253, 2009.

\bibitem{DBLP:journals/algorithmica/YamazakiBFT99}
K.~Yamazaki, H.~L. Bodlaender, B.~de~Fluiter, and D.~M. Thilikos.
\newblock Isomorphism for graphs of bounded distance width.
\newblock {\em Algorithmica}, 24(2):105--127, 1999.

\end{thebibliography}

\ifhasappendix
\newpage
\begin{appendix}
\noindent{\huge{\textbf{Appendix}}}

\medskip 

The appendix contains proofs of various lemmas and theorems from the main text. It also contains lemmas required for these proofs and the proofs of these lemmas. At the end there is also a separate list of references that only appear in the appendix.

\printproofs
\end{appendix}
\bibliographystyleapp{abbrv}
\bibliographyapp{main}
\fi

\end{document}